\documentclass[3p]{elsarticle}
\usepackage{float}
\usepackage{graphicx}
\usepackage{color}
\usepackage{amsmath,amssymb,verbatim,bm}
\usepackage{ctable} 
\usepackage{url}
\usepackage{bm}
\usepackage{algorithm}
\usepackage{algpseudocode}
\usepackage[english]{babel}

\usepackage{mdframed}
\newmdenv[
leftmargin = 0pt,
innerleftmargin = 1em,
innertopmargin = 0pt,
innerbottommargin = 0pt,
innerrightmargin = 0pt,
rightmargin = 0pt,
linewidth = 3pt,
topline = false,
rightline = false,
bottomline = false,
linecolor=red,
]{leftbar}

\definecolor{marine}{RGB}{85,100,255}
\definecolor{pomegrenate}{RGB}{255,18,94}

\usepackage{xspace}
\usepackage{xparse}

\DeclareRobustCommand{\qed}{$\square$}
\DeclareMathOperator{\theargmin}{argmin}
\DeclareMathOperator{\BERN}{Bernoulli}
\DeclareMathOperator{\variance}{var}
\newcommand{\SBM}{{\mathop{\scriptscriptstyle \mathrm{SBM}}}}
\newcommand{\F}{\mathop{\mathfrak{F}}}

\DeclareMathOperator{\M}{\mathbb{M}}
\DeclareMathOperator{\N}{\mathbb{N}}
\DeclareMathOperator{\R}{\mathbb{R}}

\DeclareMathOperator{\cM}{\mathcal{M}}

\newcommand{\bA}{\bm{A}}
\newcommand{\bE}{\bm{E}}

\newcommand{\cG}{\mathcal{G}}
\newcommand{\cH}{\mathcal{H}}
\newcommand{\cK}{\mathcal{K}}

\newcommand{\bp}{\bm{p}}
\newcommand{\bs}{\bm{s}}
\newcommand{\blamb}{\bm{\lambda}}

\DeclareRobustCommand{\argmin}[1]{\underset{#1}{\theargmin}\mspace{4mu}}
\DeclareRobustCommand{\bern}[1]{\BERN \left(#1\right)}
\DeclareRobustCommand{\var}[1]{\variance\left[#1\right]}

\NewDocumentCommand \E { m o}
{
  \IfNoValueTF {#2} {\mathbb{E}\left[#1\right]}{\mathbb{E}_{#2}\mspace{-4mu}\left[#1\right] }
}
\NewDocumentCommand \prob { m o}
{
  \IfNoValueTF {#2} {\mathbb{P}\left(#1\right)}{\mathbb{P}_{#2} \left(#1\right)}
}

\newtheorem{conjecture}{Conjecture}

\newtheorem{definition}{Definition}

\newtheorem{example}{Example}

\newtheorem{notation}{Notation}

\newtheorem{proof}{Proof of Theorem}
\newtheorem{remark}{Remark}
\newtheorem{theorem}{Theorem}

\newcommand{\ER}{Erd\H{o}s-R\'enyi\xspace}

\begin{document}
\begin{frontmatter}

  \title{Approximate Fr\'echet Mean for Data Sets of Sparse Graphs}
  \author{Daniel Ferguson}
  \author{Fran\c{c}ois G. Meyer\fnref{fnt2}}
  \address{Applied Mathematics, University of Colorado at Boulder, Boulder CO 80305}
  \fntext[fnt2]{Corresponding author: fmeyer@colorado.edu}

  \begin{abstract}
    To characterize the location (mean, median) of a set of graphs, one needs a notion of centrality that is adapted to metric
      spaces, since graph sets are not Euclidean spaces. A standard approach is to consider the Fr\'echet mean. In this work, we
      equip a set of graph with the pseudometric  defined by the $\ell_2$ norm between the eigenvalues of their respective
      adjacency matrix . Unlike the edit distance, this pseudometric reveals structural changes at multiple scales, and is well
      adapted to studying various statistical problems on sets of graphs. We describe an algorithm to compute an approximation to
      the Fr\'echet mean of a set of undirected unweighted graphs with a fixed size.
  \end{abstract}


\end{frontmatter}

\section{Introduction}
Machine learning from a set of data almost always requires some notion of average. Algorithms for clustering, classification,
and linear regression all utilize the average value of the data set \cite{HTF08}. When the distance is induced by a norm, then
the mean is a simple algebraic operation. If the data lie on a Riemannian manifold, equipped with a metric, then one can extend
the notion of mean with the concept of Fr\'echet mean \cite {pennec06}. In fact the concept of Fr\'echet mean only requires that
a (pseudo)metric between points be defined, and therefore one can consider the Fr\'echet mean of a set in a pseudometric space
\cite{frechet48}.  Not surprisingly, many machine learning algorithms, which  were developed for Euclidean spaces, can be extended
to use the Fr\'echet mean. The purpose of this paper is to solve the nontrivial problem of determining the Fr\'echet mean for
data sets of graphs when the pseudometric is the $\ell_2$ distance between the eigenvalues of the adjacency matrix.

In this work we consider a set of simple graphs with $n$ vertices. The graphs are considered to be sparse, in the
  sense that the edge density satisfies,
  \begin{equation}
    \rho_n = \Omega \left(\frac{\ln^3(n)}{n} \right).
  \end{equation}
  Because most real world networks are sparse, this constraint is realistic.  We additionally note that the vertex set must be
sufficiently large and that the technique introduced in this paper will perform poorly for sets of small graphs.

Our line of attack involves the following two intermediate results: (1) the Fr\'echet mean of a set of sparse graphs can be
approximated within any precision by a stochastic block model; (2) given a sequence of eigenvalues of an adjacency matrix, one
can recover the stochastic block model whose spectrum matches these target eigenvalues. We prove various error bounds and
convergence results for our algorithm and validate the theory with several experiments. The paper is structured as follows:
section \ref{sec:Not} introduces the notations used to refer to graphs and sets of graphs as well as defining precisely what is
meant by a random graph and a stochastic block model ensemble. Section \ref{sec:FM&EFM} defines the Fr\'echet mean problem, its
empirical alternative, and introduces the theorems necessary for our solution. Section \ref{sec:Reg} describes  briefly how the
Fr\'echet mean applies to regression for graph valued data sets and section \ref{sec:OptParam} introduces a numerical method to
find the Fr\'echet mean. Section \ref{sec:Exp} serves to experimentally validate the theory of the previous sections. We leave
all proofs of our results to the appendix.

\section{State of the Art}
We consider a set of undirected unweighted graphs of fixed size $n$, $\cG$, wherein we define a distance . To characterize the
location (mean, median) of the set $\cG$, we need a notion of centrality that is adapted to metric spaces, since graph sets are
not Euclidean spaces. A standard approach is to consider the Fr\'echet sample mean, and the Fr\'echet total sample variance.

The choice of metric is crucial to the computation of the Fr\'echet mean, since each metric induces a different mean graph.  The
Fr\'echet mean of graphs has been studied in the context where the distance is the edit distance (e.g.,
\cite{BNY20,ginestet12,jain09,jain16,jiang01} and references therein).  The edit distance reflects small scale  changes in the
graphs and therefore the Fr\'echet mean will be sensitive to the fine structural distinctions between graphs. Effectively,
the Fr\'echet mean with respect to the edit distance can then be interpreted as an average of the fine structures in
the observed graphs as measured by this distance.

In this paper, we consider that the fine scale, which is defined by the local connectivity at the level of each vertex, may be
intrinsically random.  The quantification of such random fluctuations is uninformative when comparing graphs.  We prefer to use
a distance that can detect larger scale patterns of connectivity that happen at multiple scales.

The adjacency spectral distance, which is just the $\ell_2$ norm of the difference between the spectra of the adjacency matrices
of the two graphs of interest \cite{Wilson2008}, exhibits good performance when comparing various types of graphs
\cite{meyer20c}, making it a reliable choice for a wide range of problems. Spectral distances also exhibit practical advantages,
as they can inherently compare graphs of different sizes and can compare graphs without known vertex correspondence (see e.g.,
\cite{ferrer10,MFA05} and references therein). The adjacency spectrum in particular is well-understood, and is perhaps the most
frequently studied graph spectrum \cite{Farkas2001, Flaxman2003}.

In practice, it is often the case that only the first $c$ eigenvalues are compared, where $c\ll n$. We still refer to such
truncated spectral distances as spectral distances. Comparison using the first $c$ eigenvalues for small $c$ allows one to focus
on the community structure of the graph, while ignoring the local structure of the graph \cite{Lee2014}. Inclusion of the
highest-$c$ eigenvalues allows one to discern local features as well as global. This flexibility allows the user to target the
particular scale at which she wishes to study the graph, and is a significant advantage of the spectral distances.

Instead of solving the minimization problem associated with the computation of the Fr\'echet mean in the original set $\cG$, the
authors in \cite{ferrer10} suggest to embed the graphs in Euclidean space, wherein they can trivially find the mean of the
set. Because the embedding in \cite{ferrer10} is not an isometry, there is no guarantee that the inverse of the average embedded
graphs be equal to the Fr\'echet mean. Furthermore, the inverse embedding may not be available in closed form.  In the case of
simple graphs, the Laplacian matrix of the graph uniquely characterizes the graph. The authors in \cite{ginestet17} define the
mean of a set of graphs using the Fr\'echet sample mean (computed on the manifold defined by the cone of symmetric positive
semi-definite matrices) of the respective Laplacian matrices.
\section{Notation
  \label{sec:Not}}
We denote by $G = (V,E)$ a graph with vertex set $V = \lbrace 1,2,...,n \rbrace$ and edge set $E \subset V \times V$. For
vertices $i,j \in V$ an edge exists between them if the pair $(i,j) \in E$. The size of a graph is called
$n = \vert V \vert$ and the number of edges is $m = \vert E \vert$. The density of a graph is called
$\rho_n = \frac{m}{n(n-1)/2}$. For any $n > 100$, we say a graph is sparse when

\begin{equation}
\rho_n = \Omega \left(\frac{\ln^3(n)}{n} \right).
\end{equation} 

The matrix $\bA$ is the adjacency matrix of the graph and is defined as
\begin{align*}
  \bA_{ij} = 
  \begin{cases}
    1 \quad \text{if }(i,j) \in E,\\
    0 \quad \text{else.}
  \end{cases}
\end{align*}
  We define the function $\sigma$ to be the mapping from the set of $n \times n$ adjacency matrices (square, symmetric matrices with zero entries
  on the diagonal), $\M_{n \times n}$ to $\R^n$ that    assigns to an adjacency matrix the vector of its $n$ sorted eigenvalues,
  \begin{align}
    \sigma:        \M_{n\times n} & \longrightarrow \R^n,\\
    \bA & \longmapsto \blamb = [\lambda_1,\ldots,\lambda_n],
  \end{align}
  where $\lambda_1 \geq \ldots \geq \lambda_n$ Because we often consider the $c$ largest eigenvalue of the
  adjacency matrix $\bA$, we define the mapping to the truncated spectrum as $\sigma_c$ ,
  \begin{align}
    \sigma_c:        \M_{n\times n} & \longrightarrow \R^c,\\
    \bA & \longmapsto \blamb_c = [\lambda_1,\ldots,\lambda_c].
  \end{align}
We write $G \tilde = G'$ when two graphs are isomorphic. Two graphs are isomorphic if and only if there exists a permutation
matrix $\bm P$ such that $\bA' = \bm P^T \bA \bm P$.  

\begin{definition}[Adjacency spectral pseudometric] We define the adjacency spectral pseudometric as the $\ell_2$ norm between
  the spectra of the respective adjacency matrices,
  \begin{align} 
    d_A(G,G') = ||\sigma(G)-\sigma(G')||_2. \label{distance}
  \end{align}
  The pseudometric $d_A$ satisfies the symmetry and triangle inequality axioms, but not the identity axiom. Instead, $d_A$
      satisfies the reflexivity axiom
      \begin{equation*}
        d_A(G,G) = 0, \quad \forall G \in \cG.
      \end{equation*}
      We define the truncated adjacency spectral pseudometric as
    \begin{align}
      d_{A_c}(G,G') = ||\sigma_c(G)-\sigma_c(G')||_2. \label{distance-trunc}
    \end{align}
  \end{definition}
  
  \begin{definition}[Set of graphs and sparse graphs]
    We denote by  $\cG$  the set of all simple graphs on $n$ nodes. \\

    \noindent Furthermore, we denote by $\cG_s \subset \cG$ the subset of sparse graphs for which the edge density satisfies
    \begin{equation}
      \rho (G)    = \frac{2m}{n(n-1)}= \Omega\left(\frac{\ln^3(n)}{n}\right) \label{sparsity}
    \end{equation}
    where $m$ is the number of edges of $G$.
  \end{definition}

\subsection{Random Graphs} 
We denote by $\cM (\cG)$ the space  of probability measures on $\cG$. In this work, when we talk about a measure we  always
mean a probability measure.
  \begin{definition}[Set of random graphs associated with a measure $\mu$]
    Let  $\cH$ be a subset of $\cG$, and $\mu$ a probability measure defined on $\cH$. We can extend $\mu$ to $\cG$,
    such that for any $\mathcal{A} \subset \cG$, we define $\mu(\mathcal{A}) = \mu (\mathcal{A} \cap \cH)$.  We define the set
    of random graphs distributed according  to $\mu$ to be the probability space $\left(\cH, \mu\right)$.
  \end{definition}
  \begin{remark}
    In this paper, the $\sigma$-field associated with the $\left(\cH, \mu\right)$ will always be the power set of $\cH$.
  \end{remark}
  This definition allows to recast various ensemble of random graphs (e.g., \ER, inhomogeneous \ER, stochastic block models,
  etc) using a unique notation.
\subsubsection{Kernel Probability Measures
  \label{subsec:kernProbMeas}}
  \noindent Let $\left\{ \xi_i \right\}_{i=1}^n$ be the sequence of equispaced points  in the interval $[0,1]$, $\xi_i  = i/n$.
  \begin{definition}[kernel probability measure]
    A probability measure $\mu \in \cM (\cG)$ is called a kernel probability measure if there exist a subset $\cH \subset \cG$
    and a  function $f$,
    \begin{equation}
      f: [0,1]\times[0,1] \mapsto (0,1),
    \end{equation}
    such that $f(x,y) = f(1-y,1-x)$, and such that 
    \begin{equation}
      \forall G \in \cH, \text{with adjacency matrix}\; \bA=\left(a_{ij}\right), 
      \mu\left( \left\{ \bA \right \}\right) = \mspace{-12mu} \prod_{1 \leq i < j \leq n}\mspace{-12mu}\prob{a_{ij}} = \mspace{-12mu} \prod_{1 \leq i < j \leq
        n}\mspace{-12mu} \bern{f(\xi_i,\xi_j)}.
    \end{equation}
    The function $f$ is called a kernel of $\mu$.
\end{definition}

\begin{definition}
 We denote by $\cK$ the set of all kernel functions,
\begin{equation}
  \cK = \left\{ f\vert f: [0,1]^2 \mapsto (0,1); f (x,y) = f(1-y,1-x)  \right\}.
\end{equation}
\end{definition}
\noindent We note that given the sequence $\left \{\xi_i\right \}\; 1 \leq i \leq n$ and the measure $\mu$, the kernel $f$ forms an
  equivalence class of functions, characterized by their values on the grid $\left \{ (\xi_i, \xi_i)\; 1 \leq i,j \leq n \right\}$.

\begin{remark}
  Many definitions of random graphs allow for $\lbrace \xi_i \rbrace$ to be a random sample from some probability density
  function on $[0,1]$ (e.g., \cite{borgs19,le18,janson13, olhede14}, and references therein) With the right distribution defined
  on $[0,1]$, many of the results would be identical when taking an expectation over the random sample $\lbrace \xi_i
  \rbrace$. The advantage to specifying the points as an equispaced grid allows a greater level of control when specifying
  certain properties of the kernels. For example, we will be able to guarantee the number of nodes in a community when
  considering kernels of stochastic block models, which we define in the following section.
\end{remark}    

\begin{notation}
We denote by $G_{\mu}$ a random realization of a graph $G\in \left(\cG,\mu\right)$. Sometimes, we use the notation  $G(n,f)$ if
$\mu$ is kernel probability measure. This notation  generalizes the notation of the classic \ER random graphs of
$G(n,p)$ where taking $f(x,y) = p$ is a possible kernel as in \cite{J10}.
\end{notation}

\begin{notation}
Given a measurable function $\varphi$ defined on the probability space $\left(\cG,\mu\right)$, we denote by $\E{\varphi}[\mu]$
the expected value $\E{\varphi(G)}$, when $G$ is distributed according to $\mu$.
\end{notation}
\begin{definition}
  Given a kernel probability measure $\mu$ with kernel $f$ we denote by 
  \begin{equation}
    \E{\rho_n}[\mu] = \frac{\sum_{i>j} f(\xi_i,\xi_j)}{n(n-1)/2},
  \end{equation}
  the expected density of the measure $\mu$ on the grid $\left \{ (\xi_i, \xi_i)\; 1 \leq i,j \leq n \right\}$.
\end{definition}

\subsubsection{Stochastic Block Models
  \label{subsec:sbm}}
\noindent The stochastic block model \cite{abbe17} plays an important role in this work. We review the specific features of this model using the
  notations that were defined in the previous paragraphs. The key aspects of the model are: the geometry of the blocks, the
  within-community edges densities, and the across-community edge density.

  We assume that there exists a large constant $C$ that provides an upper bound on the number of communities; the actual number of non
  empty communities is denoted by $c \leq C$. \\
  
  {\noindent \bf The geometry} of the stochastic block model is encoded using the relative sizes of the communities. We denote by
$\bs \in [0,1]^c$ the vector of relative sizes of each of the $c$ blocks. We have $0 < s_k \leq 1$, and $\sum_{k=1}^c s_k  = 1.$\\

\noindent {\bf The edge density} within a non empty block $k$, is denoted by $p_k$. We can concatenate the $p_k$ into a vector
$\bp = \begin{bmatrix} p_1 \cdots p_c \end{bmatrix}$, which describes the within-block edges densities.  Finally, we denote by
$q$ the across-community edge density.

We can parameterize a stochastic block model using one representative of the equivalence class of kernel, $f$. We simply
consider the function $f$, which is piecewise constant over the blocks, and is defined by
\begin{align}
  f: [0,1] \times [0,1] & \longrightarrow (0,1)\\
  (x,y)& \longmapsto
         \begin{cases}
         p_k & \text{if} \quad \sum_{j=1}^{k-1}s_k  \leq x <  \sum_{j=1}^{k} s_k, \; \text{and} \quad 
         \sum_{j=1}^{k-1}s_k  \leq y <  \sum_{j=1}^{k} s_k,\\
         q & \text{otherwise.}
       \end{cases}
\end{align}
This piecewise constant function is called the canonical kernel associated to the block model with measure
$\mu$ (see, e.g. Fig.~\ref{fig:Ex-f}), and we denote it by $f(x,y,\bp,q,\bs)$. 

\begin{example}
  Given $\bm s = \begin{bmatrix} 1/2 & 1/4 & 1/4 & 0  \cdots \end{bmatrix}^T$ the values of $f(x,y; \bm p, q, \bm s)$ in the
  unit square are shown in Fig.~\ref{fig:Ex-f}
\end{example}
Given the edges densities $(\bp,q)$, a block geometry $\bs$, and a graph size $n$, we can generate a random realization of
  the stochastic block model with kernel $f(x,y,\bp,q,\bs)$ by sampling the entries of its adjacency matrix $\bA=\left
    (a_{ij}\right)$ according to, 
\begin{equation}
  a_{i,j} \sim \bern{f(\xi_i,\xi_j; \bm p, q, \bm s}, \; a_{ji} = a_{ij}, \; \text{and} \; a_{ii} = 0, \quad 1 \leq i < j \leq n.
\end{equation}

\noindent Using the generalized \ER notation a random graph from the stochastic block model is denoted by
\begin{equation}
G(n,f(x,y;\bm p, q, \bm s)).
\end{equation}

\begin{definition}
We denote by $\cK_{\SBM}$ the set of piecewise constant functions on $[0,1]\times [0,1]$ that are canonical kernels of
stochastic block models,
\begin{equation}
\cK_{\SBM} = \left\{ f(x,y; \bm p, q, \bm s) ; c \in \N, \bp \in [0,1]^c, q \in [0,1], \bs \in [0,1]^c, \sum_{k=1}^c s_k
  = 1\right\}
\end{equation}
\end{definition}
\begin{figure}[H]
\centerline{
  \includegraphics[width = 0.4\textwidth]{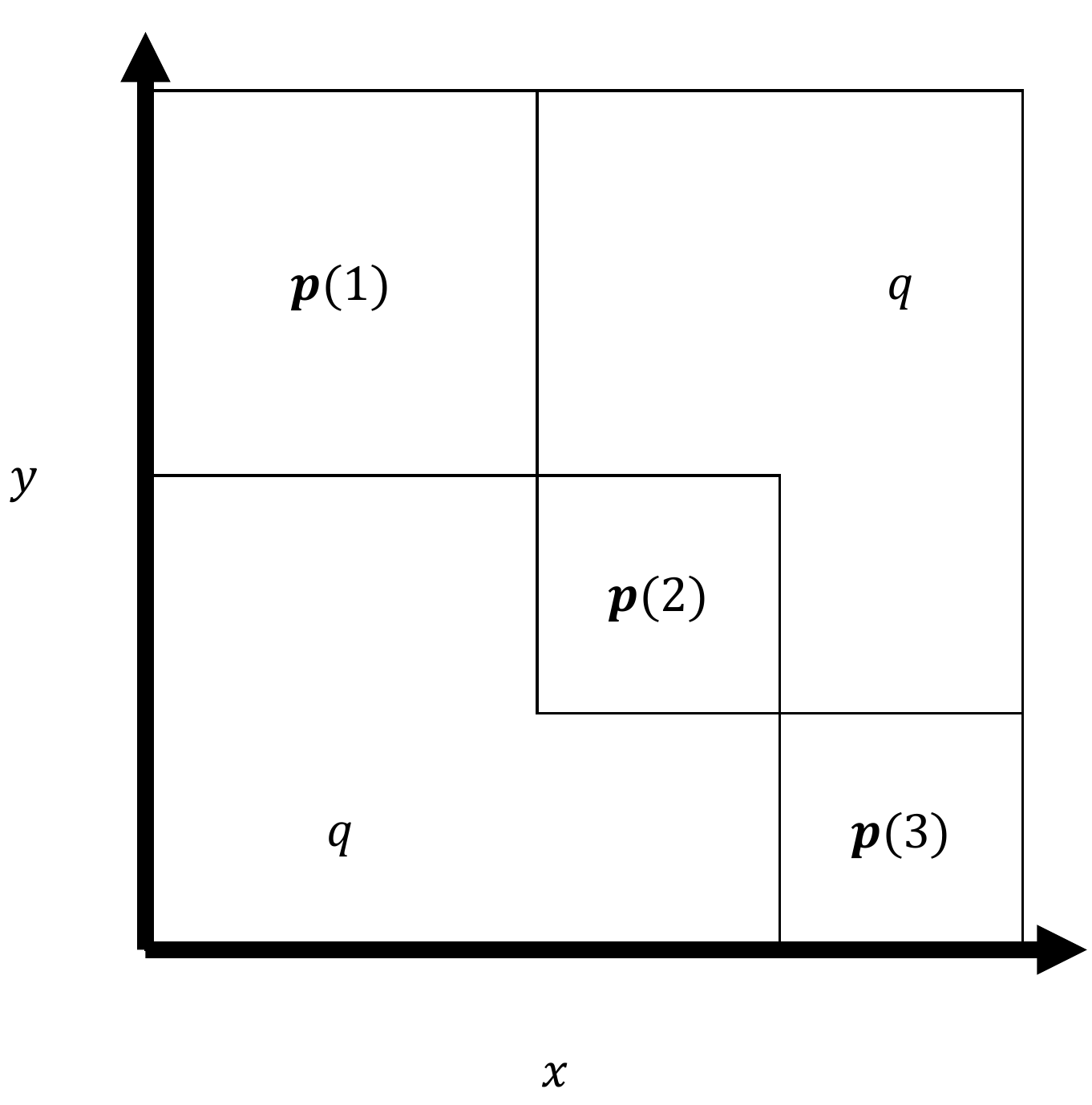}
  }
  \caption{Example $f(x,y; \bm p, q, \bm s)$}
  \label{fig:Ex-f}
\end{figure}
Defining the ensemble in this way allows for the smooth introduction of new communities by allowing the vector $\bs$ to
continuously increase in size. Additionally, defining the parameters independently of the number of nodes naturally allows
for large graph limits to be explored.

For a fixed $n$, and a canonical kernel $f_{\SBM} \in K_{\SBM}$, there exists a unique induced probability measure, which we
denote by $\mu_{\SBM}(\bm p,q,\bs, n)$. 
\begin{definition}
The set of all probability measures induced by the set $K_{\SBM}$ is denoted by
\begin{equation}
  \cM_{\SBM}(\cG) = \lbrace \mu_{\SBM} \vert \mu_{\SBM} \text{ is induced by } f_{\SBM} \in K_{\SBM}\rbrace.
  \end{equation}
\end{definition}
Sometimes we suppress the parameters and write $\mu_{\SBM} \in \cM_{\SBM}(\cG)$. 

\section{The Fr\'echet mean and Empirical Fr\'echet Mean
  \label{sec:FM&EFM}}
We consider $\cG$ the space of graphs defined on sets of $n$ vertices. We equip that space with the pseudometric defined by the
  $\ell_2$ norm between the spectra of the  respective adjacency matrices, $d_A$, (see (\ref{distance})). We consider a
  probability measure in $\mu \in \cM$ that describes the chances of obtaining a given graph when we sample $\cG$ according to
  $\mu$. Using $d_A$, we can quantify the spread of the graphs, and we can also define a notion of centrality, which gives the
  location of the average graph, according to   $\mu$.

\begin{definition}[Fr\'echet mean \cite{frechet48}]
  The Fr\'echet mean of the pseudometric space $(\cG,d)$, where $d$ is the pseudometric (\ref{distance}), equipped with
  probability measure $\mu$ is the set of graphs $G^*$ whose expected distance to $\cG$ is minimum,
  \begin{equation}
    \left\{G^* \in \cG\right\} =  \argmin{G \in \cG} \E{d^2(G,G_{\mu})}[\mu]. \label{def:FM-general}
    \end{equation}
  \end{definition}
  where $G_\mu$ is a random realization of a graph from the probability space $\left (\cG,\mu\right)$ , and the expectation
  $\E{d^2(G,G_{\mu})}[\mu]$ is computed with respect to the probability measure $\mu$.\\

  In this work, we assume that the Fr\'echet mean both exists and is unique. Therefore $\mathcal{M}^*$ is a singleton and we write
  the Fr\'echet mean as
\begin{align}
  G^* = \underset{G \in \cG}{\text{argmin }}\E{d^2(G,G_{\mu})}[\mu]. \label{def:FM} 
\end{align}

As we change $\mu$ we expect that, for a fixed $G$, $\E{d^2(G,G_{\mu})}[\mu]$ will change, and therefore the Fr\'echet mean
  $G^*$ will move inside $\cG$ for different choices of the probability measure $\mu$. $G^*$ plays the role of the center of
  mass, for the mass distribution associated with $\mu$. We make this dependency explicit by defining the Fr\'echet mean map. 

\begin{definition}[The Fr\'echet mean map]
Given the set of graphs $\cG$,  the Fr\'echet mean map assigns to a probability measure $\mu \in \cM(\cG)$ the corresponding Fr\'echet mean,
  \begin{align}
    \F(\mu) = \argmin{G \in \cG}\E{d^2(G,G_{\mu})}[\mu]. \label{eqn:FMM}
  \end{align}
\end{definition}

For a fixed $\mu$ we may write the Fr\'echet mean as $G^* = \F(\mu).$ The Fr\'echet mean map is of significant importance
throughout this paper and the map $\F$ will be referred to throughout many theorems, equations, and proofs.

In practice, the only information known about a distribution on $\cG$ comes from a sample of graphs. Therefore, we need a notion
of Fr\'echet sample mean, which is defined as follows. 
\begin{definition}[Empirical Fr\'echet mean] 
  Let $\mu$ be a probability measure on $\cG$, and let $\left\{G_i\right\}\; 1 \leq i \leq N$ be a random sample from the
  probability space $\left(\cG,\mu\right)$.  The empirical Fr\'echet mean is defined by
  \begin{align}
    G_N^* = \argmin{G \in \cG}\sum_{i=1}^Nd^2(G,G_i). \label{eqn:SFM}
  \end{align}
  We additionally assume that the empirical Fr\'echet mean exists and is unique for any given sample of graphs.  The dependence
  on $N$ is explicitly given but may be suppressed throughout the paper when it is obvious.
\end{definition}

\begin{remark}
  Because the pseudometric $d$ (see (\ref{distance})) in \eqref{eqn:SFM} is the $\ell^2$ distance between the spectra of $G$ and
  $G_i$ respectively, $G_N^*$ is the unique minimizer for \eqref{eqn:SFM}  if and only if
  \begin{equation}
    \sigma(G_N^*) = \frac{1}{N} \sum_{i=1}^N \sigma(G_i).
  \end{equation}
  There exists several computational algorithms to solve the inverse eigenvalue problem,
\begin{equation}
  G^*_N = \sigma^{-1} \left( \frac{1}{N} \sum_{i=1}^N \sigma(G_i) \right).
\end{equation}
Some recent work in this field is presented in \cite{BMB18, SK19}. The work in \cite{BMB18} involves knowledge of a graph's
adjacency matrix to generate different graphs with similar spectra. The work in \cite{SK19} is very closely related to the ideas
presented in this paper, however they utilize the spectra of the normalized Laplacian matrix and construct adjacency matrices
that will have similar spectra. This work is exciting but the algorithms presented are not currently fully supported with
theory.\\

Instead, we propose to solve the inverse eigenvalue problem over the space of sparse SBM graph models. Specifically, we
construct an SBM model, whose Fre\'echet mean is $\varepsilon$ away from $  G^*_N$, according to the truncated spectral pseudometric,
  \begin{equation}
d_{A_c}(G_N^*, G_{\SBM}^*) < \varepsilon.
\end{equation}
\end{remark}  
\section{Main Results}
We are now in position to state the main results of the paper. The next theorem constitutes the main theorem. 
\begin{theorem}[An approximation to the empirical Fr\'echet mean with respect to $d_{A_c}$]
  \label{thm:AEFM}
  ~\\
  Let $\mu \in \cM\left(\cG\right)$, and let $\left\{G_i\right\}\; 1 \leq i \leq N$ be a random sample of sparse graphs sampled 
from the probability  space $\left(\cG_s,\mu\right)$. Let 
  \begin{equation}
G_N^* = \argmin{G \in \cG}\sum_{i=1}^Nd_{A_c}^2(G,G_i)
\end{equation}
be the empirical Fr\'echet mean of the sample $\left\{G_i\right\}\; 1 \leq i \leq N$. Then, as the size $n$ of the graphs in
$\cG$ goes to infinity,
  \begin{equation}
\forall \varepsilon >0,\;    \exists \mu_{\SBM} \in  \cM_{\SBM}(\cG), \; \text{such that}\;     d_{A_c}(G^*_N,\F(\mu_{\SBM})) < \varepsilon.
\end{equation}
\end{theorem}
\begin{proof}
  The proof relies on theorem \ref{thm:Density} and conjecture \ref{cjt:EFMS}. Let $\left\{G_i\right\}\; 1 \leq i \leq N$ be a
  random sample from the probability space of sparse graph $\left(\cG_s,\mu\right)$. As a consequence of conjecture
  \ref{cjt:EFMS}, $G^*$, the Fr\'echet mean of $\left\{G_i\right\}\; 1 \leq i \leq N$ is sparse. We now apply
  theorem~\ref{thm:Density} to $G^*\in \cG_s $,
  \begin{equation}
  \forall \varepsilon > o, \; \exists \mu_{\SBM} \in  \cM_{\SBM}(\cG), \; \text{such that}\; 
    d_{A_c}(G^*,\F(\mu_{\SBM})) < \varepsilon.
  \end{equation}
  \qed
\end{proof}

The intuition for theorem \ref{thm:AEFM} comes from acknowledging that there are many distributions with the same first
moment. To utilize this idea, we need only to find a distribution $\nu$ whose first moment is the same as the empirical
Fr\'echet mean given the set of observed graphs $M$. 

The metric chosen in theorem \ref{thm:AEFM} is specifically $d_{A_c}$. This metric measures the difference in the largest $c$
eigenvalues of two graphs. It is well known that the largest $c$ eigenvalues of a stochastic block model with $c$ blocks
completely encodes the community structure. Since the community structure is a global property of the graph, we claim that other
global structures of the graphs are captured by the largest $c$ eigenvalues. As a consequence, we claim that the global
properties of the empirical Fr\'echet mean, $G_N^*$, are similar to the global properties of the approximate empirical Fr\'echet
mean, $G_{\SBM}^*$. 

Theorem \ref{thm:AEFM} shows that we need only to find the parameters of the suitable stochastic block model and compute its
Fr\'echet mean. The identification of the optimal set of parameters for the stochastic block model is a non-trivial problem. One
approach to determine  $(\bm p^*, q^*, \bs^*)$ is discussed in section \ref{sec:OptParam} though there may  be several methods
that could be used to determine these parameters that have yet to be explored. The result of theorem \ref{thm:AEFM} is one
application of the density result of theorem \ref{thm:Density}.   
\begin{theorem}[Approximation  of Fr\'echet means  in $\cG_s$ by stochastic block models]
  \label{thm:Density}
  As the size $n$ of the graphs in $\cG$  goes to infinity,
\begin{equation}
\forall \varepsilon > o, \forall G \in \cG_s , \; \exists \mu_{\SBM} \in  \cM_{\SBM}(\cG), \; \text{such that}\; 
    d_{A_c}(G,\F(\mu_{\SBM})) < \varepsilon.
  \end{equation}
\end{theorem}
\begin{proof}
  The proof is given in  \ref{proof-density}.
\end{proof}
Theorem \ref{thm:Density} states that for any sparse graph, there exists a stochastic block model who's first moment has a
similar spectrum. This density result is an extension of theorem \ref{thm:USBM2} in the appendix which is a result given in
\cite{OW14} but restated to fit the notations of this paper. The results in \cite{OW14} use a distance defined as follows, let
$f,g \in K$ be kernels for two different \textbf{kernel probability measures}, $\mu_f$ and $\mu_{g}.$ Let $\tau$ be a bijection
on the interval $[0,1]$ and let $\mathcal{T}$ be the set of all bijections on $[0,1]$. Define the distance $d_k$ as 
\begin{equation}
d_k(f,g) := \underset{\tau \in \mathcal{T}}{\text{inf }} \int \int_{(0,1)^2} \left|  f(\tau(x),\tau(y)) - g(x,y)\right|^2 dx dy.
\end{equation}
The result given in \cite{OW14} shows that for any sparse graph, there exists a kernel probability measure $\mu_k$ with kernel
function $k$ that generates the original sparse graph. It is then shown that the kernel $k$ may be approximated by a kernel
$k_{\SBM} \in K_{\SBM}$ when considering the distance $d_k$.

We are aware of two reasons the use of $d_k$ will lead to difficulties when computing the empirical Fr\'echet mean. Primarily,
the distance $d_k$ acts on a kernel rather than on a graph. As such, the kernel for a graph must be estimated. This is done in
\cite{OW14} when the graph is known, but in the case of the empirical Fr\'echet mean, the target graph whose kernel must be
found is unknown. Additionally, it is unclear how to easily compare the distance between multiple graphs, as needs to be done in
the empirical Fr\'echet mean problem since for each distance computation a specific bijection, $\tau^*$, needs to be
identified. When the graphs considered do not have node correspondence this becomes a non-trivial problem. One solution is to
assume the empirical Fr\'echet mean and the sample graphs have node correspondence but this restricts the applicability of the
theory.

The advantage to using $d_{A_c}$ is largely due to the ignorance $d_{A_c}$ has of the particular vertex labels in each
graph. Since the eigenvalues are consistent with respect to any permutation of the adjacency matrix we do not have to worry
about identifying a correct labelling of the nodes. Furthermore, since the distance $d_{A_c}$ acts directly on the graphs, there
is no intermediate step of identifying a kernel that best represents a given graph as in \cite{OW14}.

A necessary result for theorem \ref{thm:AEFM} is to show the empirical Fr\'echet mean meets the sparsity condition of theorem
\ref{thm:Density} so that it may be approximated by the Fr\'echet mean of a stochastic block model. Conjecture \ref{cjt:EFMS}
gives a states when these conditions are met. 
\begin{conjecture}[The empirical Fr\'echet mean of sparse graphs is sparse]
  \label{cjt:EFMS}
  Let $\mu$ be a probability measure in $\cM\left(\cG\right)$. Let $\left\{ G_i \right\} 1\leq i \leq N$ be a sample of $N$
  sparse graphs from $\left(\cG_s,\mu\right)$. We consider the empirical Fr\'echet mean computed according to the metric
  $d_{A_c}$ ,
  \begin{equation}
G_N^* = \argmin{G \in \cG}\sum_{i=1}^N d_{A}^2(G,G_i)
\end{equation}
Then, in the limit of large graph size $G_N^*$ is sparse.
\end{conjecture}
  \begin{remark}
    In the case where the distance is the edit distance, then there is a simple constructive proof. In this work, we use a
    norm based on the eigenvalues of the adjacency matrices, $d_{A_c}$, and we do not have a proof at the time of writing.  While
    there are several relationships between the degrees and the eigenvalues, we do not have a fine characterization of the
    spectra of sparse graphs. Using the density of stochastic block models in the space of sparse graphs, we have tried to prove
    a simpler result: the Fr\'echet mean of a sample of sparse stochastic block models is sparse. Unfortunately, the state of
    the art techniques to compute the spectrum of the Fr\'echet mean of a stochastic block model relies on numerical algorithms
    to compute eigenvalues and eigenvectors in addition to a root finding procedure, which we have not been able to use to prove
    that the density of the Fr\'echet mean is in fact sparse.
\end{remark}
\section{Identifying $(\bm p^*,q^*, \bs^*)$ for theorem 1
\label{sec:OptParam}}
While the result of theorem \ref{thm:AEFM} is interesting in its own right, this section provides a numerical method to
determine the optimal set of parameters $(\bm p^*,q^*, \bs^*)$ so the theorem may be used in practice. First we equate the
search over the space of graphs with a search over the set of distributions. Let $\mu$ be a probability measure on $\cG$. Let $M
= \lbrace G_i \rbrace_{i = 1}^N$ be an iid sample distributed according to $\mu$. Let $G_N^*$ be the empirical Fr\'echet mean
and let $m^*$ be the objective value of the empirical Fr\'echet mean problem, equation (\ref{eqn:SFM}), evaluated at $G_N^*$. 

Since there exists many distributions $\nu^* \in \cM (\cG)$ with the property $\F(\nu^*) = G_N^*$, we may find any distribution with this property. This observation allows us to rewrite equation (\ref{eqn:SFM}) in the following way. 
\begin{align}
  m^*     &= \underset{G \in \cG}{\text{min }}\sum_{i = 1}^N d^2_{A_c}(G,G_i) \label{eqn:FM-D-orig}\\ 
          &= \sum_{i=1}^N d_{A_c}^2(G_N^*,G_i)\\
          &= \sum_{i=1}^N d_{A_c}^2(\F(\nu^*),G_i)\\
          &= \underset{\nu \in \cM (\cG)}{\text{min }}\sum_{i = 1}^N d^2_{A_c}(\F(\nu),G_i) \label{eqn:FM-D}
\end{align}
At this point it is worth briefly mentioning this re-characterization of the optimization procedure is not unique to the metric space of graphs. This change of space, from the metric space to the space of probability distributions, can be applied to any metric space though whether this change is helpful in solving the optimization problem is unknown.

By theorem \ref{thm:AEFM}, if $G_N^*$ is sparse then we may approximate it by taking the Fr\'echet mean of a suitable stochastic
block model. Therefore we may restrict the search space in equation (\ref{eqn:FM-D}) from $\cM (\cG)$ to $\cM_{\SBM}(\cG)$ if we
allow for an approximate solution. The approximate solution of equation (\ref{eqn:FM-D}) is the solution to the following
minimization problem. 
\begin{align}
  m_{\SBM}^* = \underset{\nu_{\SBM} \in \cM_{\SBM}(\cG)}{\text{min }}\sum_{i = 1}^N d^2_{A_c}(\F(\nu_{\SBM}),G_i). \label{eqn:FM-D1}
\end{align}
By taking the argmin of equation (\ref{eqn:FM-D1}) we identify the correct stochastic block model, 
\begin{align}
  \nu_{\SBM}^* = \underset{\nu_{\SBM} \in \cM_{\SBM}(\cG)}{\text{argmin }}\sum_{i = 1}^N d^2_{A_c}(\F(\nu_{\SBM}),G_i) \label{eqn:FM-D2}
\end{align}
and only need to evaluate $\F(\nu_{\SBM}^*)$ to determine the approximate empirical Fr\'echet mean,
\begin{equation}
G_{\SBM}^* = \F(\nu_{\SBM}^*).
\end{equation}

The change of the space from $\cG$ to $\cM (\cG)$ is motivated by ideas present in \cite{AGS08} which shows that the space of probability measures has curvature. This curvature is essential for searching $\cG$ in a principled way. Rather than rely on the Wasserstein metric between distributions, we restricted $\cM(\cG)$ to the class of stochastic block models, $\cM_{\SBM}(\cG)$, where a Euclidean distance between the parameters is sufficient.

By restricting to the subset of distributions associated with stochastic block models, $\cM_{\SBM}(\cG)$, the search over the distributions is equivalent to a search over the parameters $\bm p, q, \text{ and } \bs$. The equivalent problem to equation (\ref{eqn:FM-D2}) in terms of the parameters is
\begin{align}
  (\bm p^*, q^*, \bs^*) = \underset{\bm p, q, \bs}{\text{argmin }} \sum_{i=1}^N d_{A_c}^2(\F(\nu(\bm p, q, \bs)), G_i). \label{eqn:AFM-P}
\end{align}

Before implementing any optimization procedure we simplify the objective by utilizing some theory. Motivated by the practical implementation of this work, the theory in \cite{OW14} gives a method to determine the entries of $\bs^*$ for a finite $n$ and known number of communities $c$. In algorithm \ref{alg:DetC} we outline a heuristic approach to determine $c$ which, when coupled with a finite graph, determines the entries of $\bs^*$. It is worth noting that any heuristic algorithm that estimates the number of communities in a graph is sufficient.

\begin{algorithm}[H]
  \caption{Determine $c^*$ for the approximate empirical Fr\'echet mean}
  \label{alg:DetC}
  \begin{algorithmic}[1]
    \Require Set of graphs, $M$, and integer $K$
    \State Compute the average spectrum of graphs in $M$ as $ \bm \bar{\lambda}$.
    \State Initialize $i = 0$.
    \State \textbf{Do}
    \State \hspace{0.5cm} $i=i+1$
    \State \hspace{0.5cm} Initialize $r = \bm \bar{\lambda}(i)$
    \State \hspace{0.5cm} Initialize the semi-circle probability density function, \cite{ACK15}, as $s(\lambda;r)$ where $r$ is the radius.
    \State \hspace{0.5cm} Assume $\bm \bar{\lambda}(j) \sim s(\lambda;r)$ for $j = i,...,n$. Distributed according to $s(\lambda;r)$
    \State \hspace{0.5cm} Create the pdf of the $K$ largest order statistics with a sample size $n-i$, $\lambda_{(n-i)},...,\lambda_{(n-i-K+1)}$
    \State \hspace{0.5cm} Numerically compute the expected value of the $K$ largest order statistics from the pdf $s(\lambda;r)$ with a sample size of $n-i$.
    \State \hspace{0.5cm} With sample size $n-i$ numerically compute the standard deviation of the $K$ largest order statistics, $\sigma_{n-i},...,\sigma_{n-i-K+1}$\\
    \textbf{While} $\vert \bm \bar{\lambda}(1+i) - \E{\lambda_{(n-i)}}\vert >\sigma_{n-i} \lor ... \lor \vert \bm \bar{\lambda}(K+i) - \E{\lambda_{(n-i-K)}}\vert >\sigma_{n-i-K+1}$
    \State \textbf{Return:} $c^* = i-1$
  \end{algorithmic}
\end{algorithm}
This algorithm assumes that the eigenvalues of the empirical Fr\'echet mean can be partitioned into bulk eigenvalues and extremal eigenvalues, refer to \cite{ACK15,FFHL19,T12} for further information on bulk eigenvalue distributions. The algorithm presented then detects when $K$ sequential eigenvalues are drawn from the distribution of the bulk and states that the number of extremal eigenvalues is equal to the number of communities present in the graph. We assume the shape of the bulk follows the classic semi-circle law. While the shape of the bulk is important, the crucial part of the distribution to consider in this algorithm is the shape at the edge of the bulk since this region impacts the location of the largest order statistics. Note that in practice, for any heuristic algorithm used, it may be useful to implement an upper bound on the result of the algorithm to limit the number of communities.

The knowledge of $c^*$ dictates the number of non-zero entries in both $\bs^*$ and $\bm p^*$. To further simplify the objective we restrict the class of stochastic block models considered to ones where all communities are equal sized with the exception that one community is allowed to be larger. Since the density result in theorem \ref{thm:Density} holds with this restricted ensemble, as shown in the appendix, we may now uniquely determine the entries of $\bs^*$. This restricted class of stochastic block models means that for $n < \infty$, and given $c^*$, the number of vertices may be written as
\begin{equation}
n = c^* w + r, \qquad w,r \in \mathbb{N}.
\end{equation}

We then set the non-zero entries of $\bs^*$ as 
\begin{align}
  \bs^*(1) &= \frac w n + \frac r n\\
  \bs^*(i) &= \frac{w}{n} \quad i = 2,...,c^*. \label{eqn:optS}
\end{align}
This reduces equation (\ref{eqn:AFM-P}) to
\begin{align}
  (\bm p^*, q^*) = \underset{\bm p, q}{\text{argmin }} \sum_{i=1}^N d_{A_c}^2(\F(\nu(\bm p, q, \bs^*)), G_i). \label{eqn:AFM-P1}
\end{align}
A further simplification to equation (\ref{eqn:AFM-P1}) is due to the following conjecture.
\begin{conjecture}
  \label{cjt:AvgSpar}
  \textbf{The empirical Fr\'echet mean with metric $d_{A_c}$ inherits the euclidean averaged density of the observed graphs} \hfill \\
  Let $\mu$ be a probability measure on $\cG$. Let $M = \lbrace G_i \rbrace_{i = 1}^N$ be an iid sample distributed according to $\mu$ such that for all $i$, $G_i \in \cG_s$. For the metric $d_{A_c}$ let 
  \begin{equation}
G_N^* = \argmin{G \in \cG}\sum_{i=1}^N[d_{A_c}^2(G,G_i)]
\end{equation}

  be the empirical Fr\'echet mean with density $\rho^*$. Let $\rho_i$ denote the density of $G_i$ and denote the average density of the observed graphs as $\bar \rho = \frac 1 N \sum_{i=1}^N \rho_i,$ then
  \begin{equation}
\rho^* = \bar \rho.
\end{equation}

\end{conjecture}
Conjecture \ref{cjt:EFMS} states the empirical Fr\'echet mean of sparse graphs is sparse. The statement here is stronger. In conjecture \ref{cjt:AvgSpar} we claim that the average density of the observed graphs is the density of the empirical Fr\'echet mean. This conjecture gives a necessary conditions on one of the parameters and so for any choice of $\bm p$, the parameter $q$ is chosen to meet the expected density requirement of conjecture \ref{cjt:AvgSpar}. These results reduce equation (\ref{eqn:AFM-P1}) to the final simplified form in equation (\ref{eqn:Conv-P}).

\subsection{Simplified objective and numerical algorithm}
We now restate the approximate Fr\'echet mean problem in its simplified form and present an algorithm to determine the solution numerically.

Let $\mu$ be a probability measure on $\cG$. Let $M = \lbrace G_i \rbrace_{i = 1}^N$ be an iid sample distributed according to $\mu$ such that 
for all $i$, $G_i \in \cG_s$. Let $\rho_i$ denote the density of each observed graph with $\bar{\rho}$ denoting the euclidean averaged density of the observed graphs. For any choice of parameters $\bm p, q, \bs$, let $f(x,y; \bm p, q, \bs)$ be the canonical kernel function and $\nu_{\SBM}$ denote the associated distribution. Let $\E{\rho \vert f}$ denote the expected density of graphs sampled from the $\nu_{\SBM}$.
\begin{align}
  \bm p^* &= \underset{\E{\rho \vert f} = \bar \rho}{\underset{\bm p}{\text{argmin }}}\sum_{i=1}^N d^2_{A_c}(\F(\nu(\bm p,q,\bs^*)),G_i) \label{eqn:Conv-P}\\
  G_{\SBM}^* &= \F(\nu(\bm p^*, q^*, \bs^*))
\end{align}
We claim equation (\ref{eqn:Conv-P}) is convex and can be minimized by taking projected gradient descent steps. As we do not have access to the derivative of the objective we instead use a simple second order centered difference numerical approximation which leads to the following algorithm to determine the empirical Fr\'echet mean.

\begin{algorithm}[H]
  \caption{Approximating the empirical Fr\'echet mean}
  \label{alg:DetFM_C}
  \begin{algorithmic}[1]
    \Require Set of graphs, $M = \lbrace G_i \rbrace_{i = 1}^N$
    \State Compute the average density $\bar{\rho}_n$ of the graphs in $M$
    \State Compute $c^*$ via algorithm \ref{alg:DetC} and determine $\bs^*$ via equation (\ref{eqn:optS}) 
    \State Initialize $\bm p$
    \State Initialize $q$ such that $\bar{\rho} = \E{\rho \vert f}$
    \While{Relative change in $\bm p$ and $q$ is large}
    \State {\hspace{1cm} Estimate the gradient of $\sum_{i = 1}^N d^2_{A_c}(\F(\nu(\bm p, q, \bs^*)),G_i)$ via centered differences using equation (26) of \cite{FFHL19} when necessary to determine the spectrum}
    \State{\hspace{1cm}Take a projected gradient descent step to update $\bm p$}
    \State {\hspace{1cm} Determine $q$ such that $\bar{\rho} = \E{\rho \vert f}$ with updated $\bm p$}
    \EndWhile
    \State {Estimate $G_N^* = \F(\nu(\bm p^*,q^*, \bs^*))$ via theorem \ref{thm:FiniteSampFM}}\\
    \textbf{Return: }$G_N^*$.
  \end{algorithmic}
\end{algorithm}

\section{Application to Regression}
\label{sec:Reg}
In this section we provide an application of the computation of the empirical Fr\'echet mean: the construction of a regression
function in the context where we observe graphs that depend on a real-valued random variable. We follow the approach described
in \cite{PM19}, and we replace the computation of the Fr\'echet mean with our algorithm. We consider the following scenario. Let
$\mu \in \cM\left( \cG \right)$, and let $T$ be a random variable with probability density $\prob{t}[T]$. We consider the 
random variable formed by the pair $G$ and $T$, distributed with the joint distribution formed by the product $\mu \times
\prob{}[T]$.  We wish to compute the regression function
\begin{equation}
\E{G \vert T = t}.
\end{equation}
The authors in \cite{PM19} propose to compute the following regression function
\begin{align}
  m(t) = \argmin{G \in \cG}\E{s(T,t)d^2(G,G_{\mu})}[\mu \times \prob{}[T]] \label{eqn:FR}
\end{align}
where the expectation in \eqref{eqn:FR} is computed jointly over $G_\mu$ distributed according to $\mu$, and $T$, distributed
according to $\prob{}[T]$, and the bilinear form $s$ is defined by
\begin{equation}
s(T,t) = 1 + (T - \E{T}) \left[ \var{T} \right] ^{-1}(t-\E{T}).
\end{equation}
The bilinear form $s(T,t)$ plays the role of a kernel, returning the location of $t$ with respect to the location ($\E{T}$) and
scale $\var{T}$) of $T$. The regression function $m(t)$ returns a kernel estimate of the linear regression function by summing
over all the possible pairs $(G_\mu,T)$. We note that the regression function returns the Fr\'echet mean $G^*$, when evaluated
at $t=\E{T}$.  

Given a finite sample, $M = \lbrace (t_i,G_i) \rbrace_{i=1}^N$ from $\mu$, we would like to estimate

The proposed estimation function in 
for any arbitrary distance $d$. The objective in equation (\ref{eqn:FR}), when considering a Euclidean metric for data in $R^d$
rather than graphs in $\cG$, computes the classic linear regression solution found from least squares. The weight function,
$S(T,t)$ computes some measure of how far the data is from the expected value and weights the observations accordingly with the
property that $\E{S(T,\E{T})} = 1$. Note that when $S(T,t) = 1,$ equation (\ref{eqn:FR}) reduces to the standard definition of
the Fr\'echet mean. The empirical estimate of equation (\ref{eqn:FR}) is the natural estimate where each unknown term is
replaced with the empirical alternative as
\begin{align}
  \hat{m}(t) = \argmin{G \in \cG} \sum_{i=1}^N s_{i,N}(t)d^2(G,G_i) \label{eqn:EFR}
\end{align}
where 
\begin{equation}
s_{i,N}(t) = 1 + (t_i - \bar{T})\hat{V}(t - \bar{T}).
\end{equation}

Here we have used $\bar{T}$ and $\hat{V}$ as the empirical estimate of the mean and variance of $T$. The objective in
(\ref{eqn:EFR}) can be interpreted as a weighted empirical Fr\'echet mean with weight function $s_{i,N}(t)$.  

Assume $G_i \in \cG_s$ and when choosing the distance function as $d_{A_c}$ that
\begin{equation}
\hat{m}(t) \in \cG_s.
\end{equation}

Then by theorem \ref{thm:AEFM}, we may estimate the values of $\hat{m}(t)$ by the Fr\'echet means of appropriate stochastic
block models. Theorem 1 says that for every $t \in [0,1]$, and for every $\varepsilon > 0$, there exists a set of parameters for
the stochastic block model dependent on $t,$  
\begin{equation}
(\bm p^*(t), q^*(t), \bs^*(t))
\end{equation}

such that in the limit of large graph size \begin{equation}
d_{A_c}(\hat{m}(t),\F(\mu_{\SBM}(\bm p^*(t), q^*(t), \bs^*(t)))) < \varepsilon.
\end{equation}

For every time $t$, we have related the value of the regression function $\hat m(t)$ with a set of parameters for the stochastic
block model where the Fr\'echet mean of the stochastic block model is close, with respect to $d_{A_c}$, to the optimal graph
$\hat m(t)$.

To evaluate the regression function $\hat m(t)$ an empirical Fr\'echet mean problem must be solved. Therefore, the computation
time for any time $t$ is equivalent to the speed at which we can determine the optimal set of parameters for a stochastic block
model. In section \ref{sec:OptParam} we discuss one approach to identify the optimal parameters but it involves a costly
optimization procedure. In coming papers we explore methods to speed up the process of determining these parameters but this
subject is out of scope for the current paper.

To measure the quality of the fit given by $\hat m (t)$ we define the error as 
\begin{equation}
e = \sum_{i=1}^N d_{A_c}^2(\hat{m}(t_i), G_i)
\end{equation}

for some distance function $d$. Incorporating the error in approximating the empirical Fr\'echet mean, the error term is then 
\begin{equation}
e_{\SBM} = \sum_{i=1}^N d^2(\F(\mu_{\SBM}(\bm p^*(t_i), q^*(t_i), \bs^*(t_i))), G_i).
\end{equation}

This is analogous to a sum of square errors for linear regression performed in Euclidean space and informs us as to the quality of the fit of the linear regression.
\section{Experiments}
\label{sec:Exp}
We illustrate the theory of the previous sections by examining experimental results using various synthetic datasets of
graphs. We first validate the consistency of the theory and then explore some limitations. Each data set consists of $N = 50$
graphs on $n = 600$ nodes. 

We consider five different iid data sets of graphs, $M_1,...,M_5$, drawn from distributions $\mu_1,..,\mu_5$ respectively. The distributions have the following high level descriptions. 
\begin{center}
  \begin{tabular}{ l l}
    $\mu_1$: &A stochastic block model\\ 
    $\mu_2$: &Distribution of dense graphs\\  
    $\mu_3$: &Variable community sizes\\
    $\mu_4$: &Small world\\
    $\mu_5$: &Barabasi-Albert
  \end{tabular}
\end{center}
For each dataset, we determine the parameters of the stochastic block model whose Fr\'echet mean is close to the empirical Fr\'echet mean and label these distributions $\nu_1,..., \nu_5$. Within each subsection we discuss the specific parameters for each distribution when applicable. All the code and data is provided at \url{https://github.com/dafe0926/approx_Graph_Frechet_Mean}.\\

\subsection{Objective Decay}

Prior to examining the quality of the estimate of the sample Fr\'echet mean, it is necessary to verify that the objective in
equation (\ref{eqn:Conv-P}) is indeed convex. While we provide no analytic result, figure \ref{fig:objDec} serves to justify
taking projected gradient descent steps minimizes the objective. Furthermore, the objective converges to zero for all sample
sets except for the sample of dense graphs from $\mu_2$ which is consistent with theorem \ref{thm:Density}.
\subsection{Consistency
\label{subsec:cons}}
To verify the consistency of our algorithm we would like to exactly recover the empirical Fr\'echet mean up to a relabeling of the nodes when the graphs in our dataset are drawn from a stochastic block model as in dataset $M_1$. In figure \ref{fig:consistencyFM} we display the adjacency matrix of an arbitrary graph from the set $M_1$ and the adjacency matrix of the estimated empirical Fr\'echet mean $G_N^*$. In figure \ref{fig:consistencyFM}, the empirical Fr\'echet mean is similar to an observed graph. This is because the stochastic block model induces a normal distribution for the extreme eigenvalues with small variance in the limit of large graph size resulting in any observation being close to the Fr\'echet mean. 

As is clear in figure \ref{fig:consistencyFM}, the community strengths of the observed graph are not aligned with the community
strengths of the empirical Fr\'echet mean. This misalignment is both the advantage and disadvantage to using the distance
$d_{A_c}$. In general, we do not expect the observations in set $M_1$ to have a consistent node labeling with the empirical
Fr\'echet mean so there is no reason to preserve the node labels of the graphs in the dataset. However, when the node labels of
the empirical Fr\'echet mean and the graphs in the dataset should align, a heuristic algorithm must be introduced to recover the
correct vertex labeling for the vertices in the empirical
\begin{figure}[H]
  \centerline{
  \includegraphics[width=0.6\textwidth]{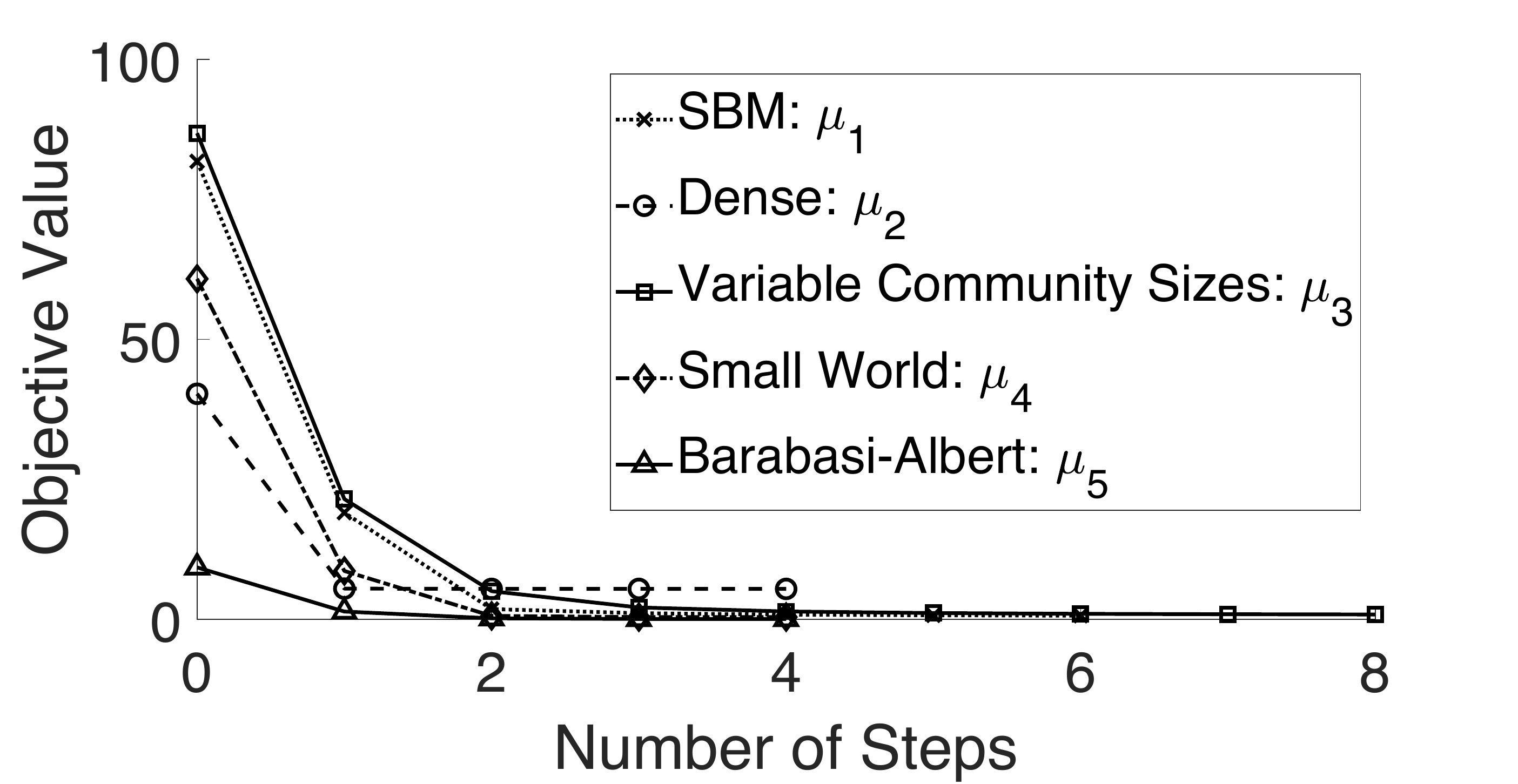}
}
  \caption{Verification of convex objective}
  \label{fig:objDec}
\end{figure}
\begin{figure}[H]
  \centerline{
    \includegraphics[width=0.6\textwidth]{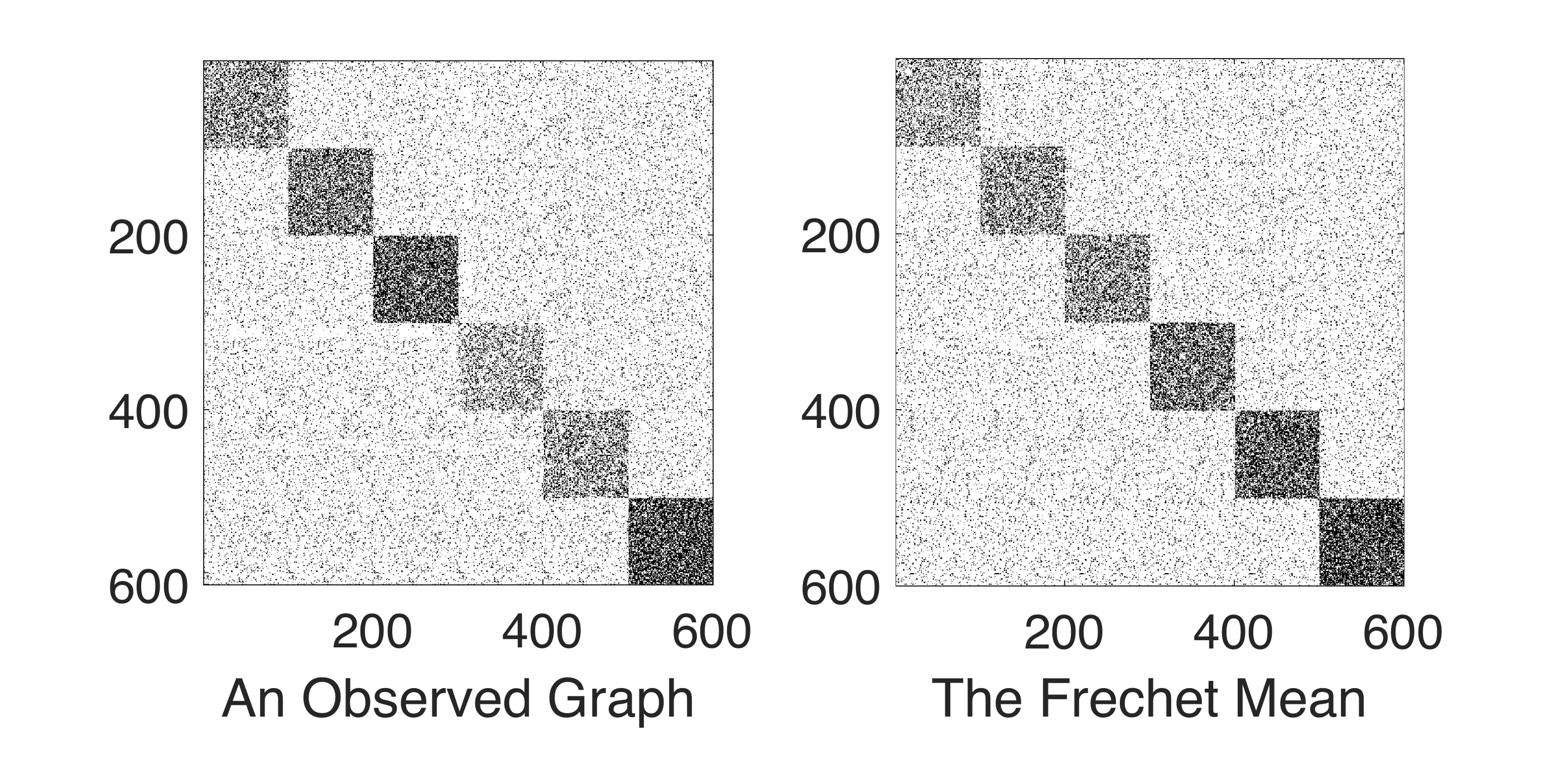}
    }
  \caption{Visual comparison between an observation from $\mu_1$ and the approximate Fr\'echet mean of $\mu_1$}
  \label{fig:consistencyFM}
\end{figure}
\noindent Fr\'echet mean. An example of such a dataset could be a temporal set in which the node labels are consistent across
all the graphs. 

Due to the mislabeling of nodes, a better graphic to visually inspect the quality of the estimate in theorem \ref{thm:AEFM} is
to compare the spectra of the observed graphs versus the spectra of the empirical Fr\'echet mean. In figure
\ref{fig:consistencyHist} below, we compute the euclidean average of the observed spectra from the graphs in $M_1$. We then
determine the Fr\'echet mean graph as described in algorithm \ref{alg:DetFM_C} and compute the histogram of its spectra. The
histogram of spectral values for random matrices is discussed throughout the works of \cite{ACK15, T12, T18,
  FFHL19,Z18}. Overlaying the average histogram from the observed set and the histogram of $G_N^*$ gives a sense as to how well
the approximation recovers the euclidean averaged eigenvalues.

Figure \ref{fig:consistencyHist} shows that by capturing the behavior of the largest eigenvalues, as guaranteed by theorem
\ref{thm:AEFM}, we exactly recover the entire distribution of the set of graphs as is shown by the alignment of the distribution
of the extremal and bulk eigenvalues. This result suggest that the largest eigenvalues of the stochastic block model completely
characterize the models behavior. 
\subsection{Dense Empirical Fr\'echet Mean
\label{subsec:dense}}
In an obvious extension of the theory, we attempt to understand the consequences that arise when given a sample of dense
graphs. The probability measure, $\mu_2$, for this section has an expected density
\begin{equation}
\E{\rho \vert \mu_2} > \frac{\ln^3(n)}{n}.
\end{equation}

In fig. \ref{fig:denseFM} we display the adjacency matrix of an arbitrary graph from the set $M_2$ and the adjacency matrix of
the estimated empirical Fr\'echet mean $G_N^*$.
\begin{figure}[H]
  \centerline{
  \includegraphics[width=0.6\textwidth]{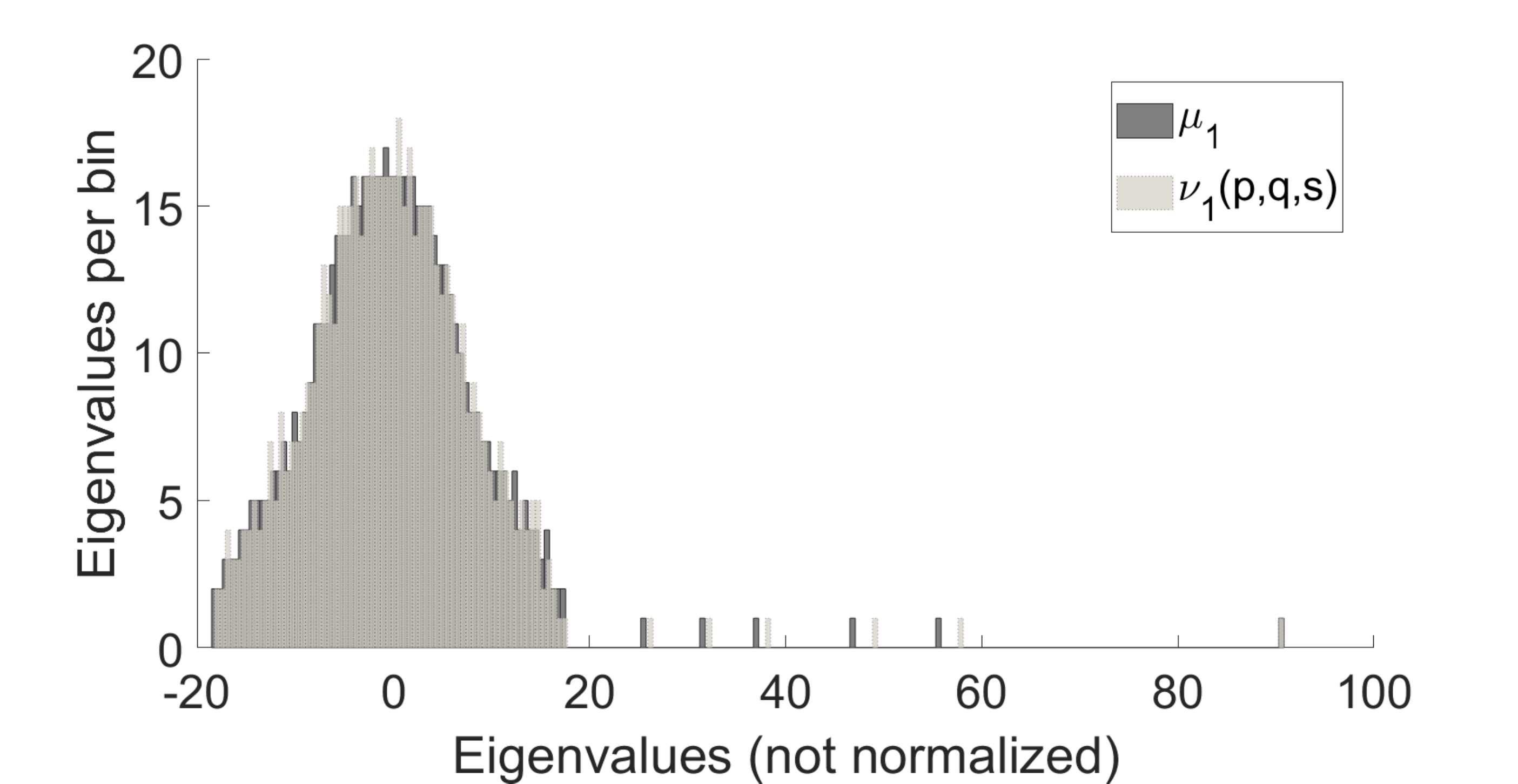}}
  \caption{Average histogram of spectra from $M_1$ overlayed by histogram of spectra from $\F(\nu_1)$}
  \label{fig:consistencyHist}
\end{figure}
\begin{figure}[H]
  \centerline{
    \includegraphics[width=0.6\textwidth]{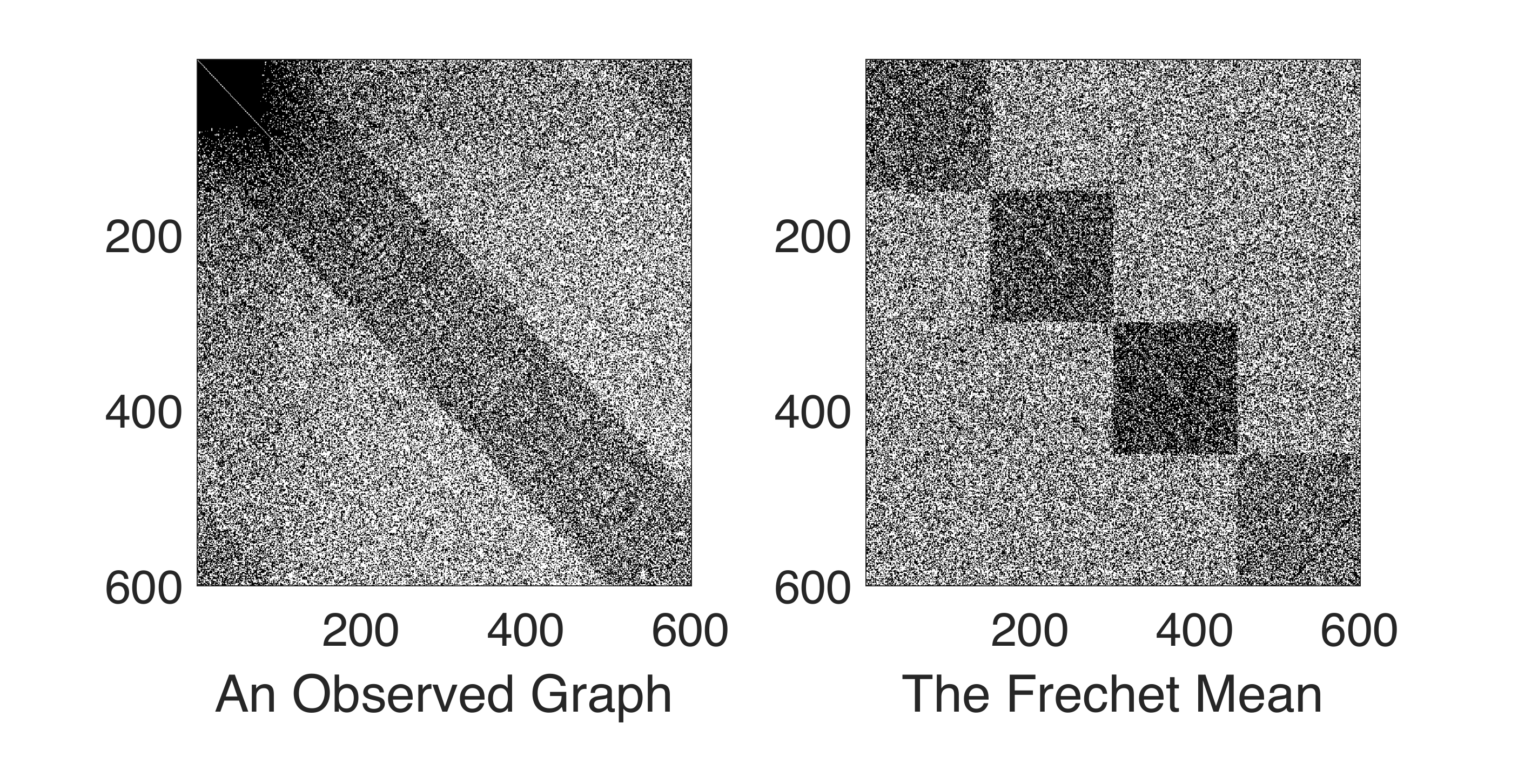}
    }
  \caption{Visual comparison between an observation from $\mu_2$ and the approximate Fr\'echet mean of $\mu_2$}
  \label{fig:denseFM}
\end{figure}
\begin{figure}[H]
  \centering
  \includegraphics[width=0.6\textwidth]{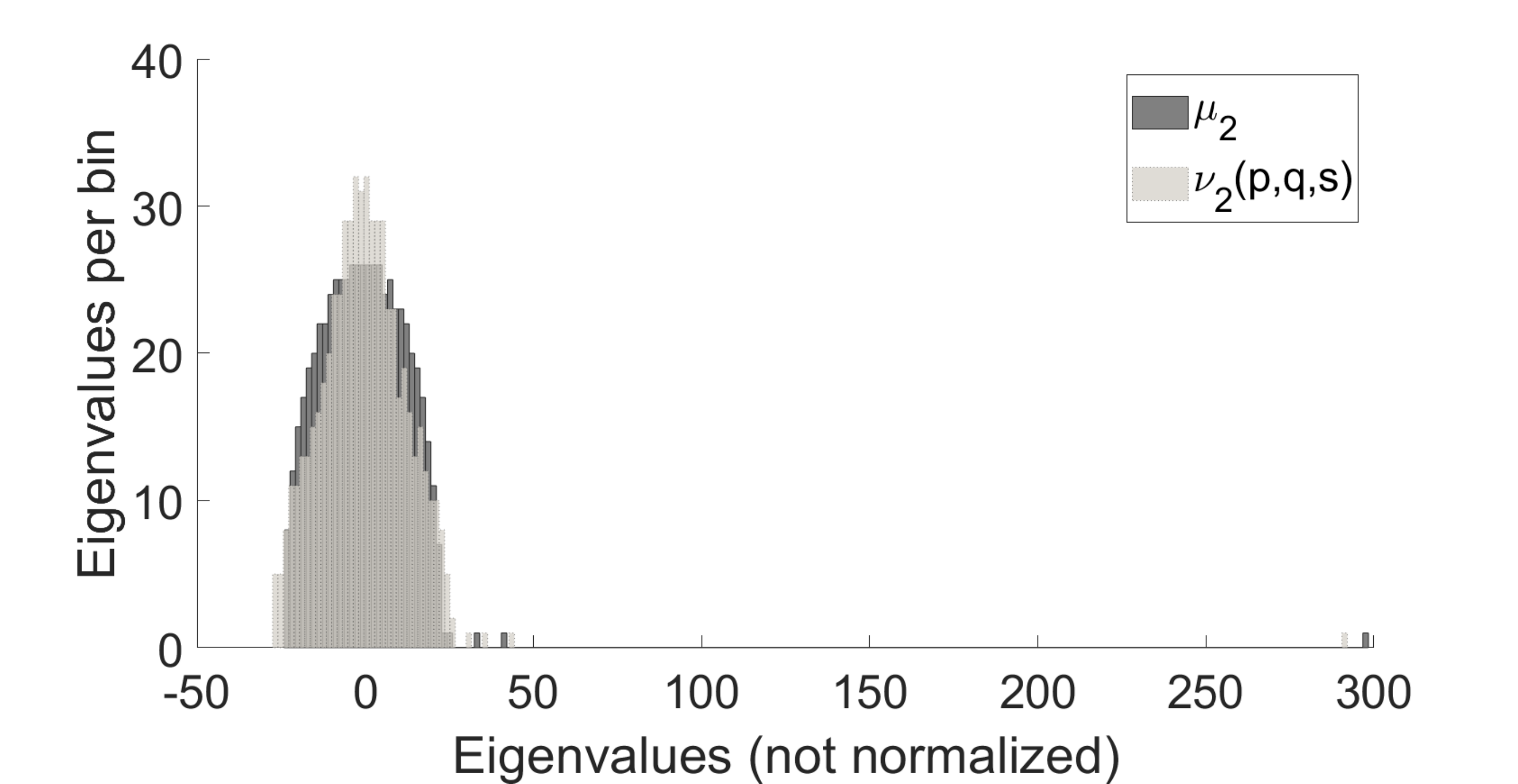}
  \caption{Average histogram of spectra from $M_2$ overlayed by histogram of spectra from $\F(\nu_2)$}
  \label{fig:denseHist}
\end{figure}
Figure \ref{fig:denseFM} illustrates  the obvious distinction between the visualization of the adjacency matrix of the empirical
Fr\'echet mean and an arbitrary adjacency matrix of a graph from $M_2$. This distinction is expected since the average of a set
need not resemble any one element of the set in theory. In addition, the only guarantee from theorem 1 is that the extremal
eigenvalues of the empirical Fr\'echet mean matches the average of the extremal eigenvalues of the graphs in the dataset. In
figure \ref{fig:denseHist} we again plot the average histograms of the observed graphs against the histogram of the empirical
Fr\'echet mean as in section \ref{subsec:cons}. In figure \ref{fig:objDec} we saw that the objective did not decay to zero so we
do not expect a perfect alignment of the extremal eigenvalues in figure \ref{fig:denseHist}.
Figure \ref{fig:denseHist} shows the misalignment of the largest eigenvalue of the empirical Fr\'echet mean from the average
largest eigenvalue of the graphs in the dataset $M_2$. This misalignment could be due to the graphs in the dataset being too
dense which is well known to be related to the magnitude of the largest eigenvalue for a graph.

\subsection{Variable Community Sizes in Stochastic Block Models}

While we introduce stochastic block models for variably sized communities, the practical applications of the theory resulted in
restricting the class of stochastic block models in our search space to those with equal sized communities with at most one
larger community. Nonetheless, a stochastic block model with variably sized communities can be approximated by a stochastic
block model with equal sized communities just as well. In this section we explore this idea. While we see clear distinctions
between the visualization of the adjacency matrices in figure \ref{fig:varSizeFM}, the alignment of the extremal eigenvalues
remains accurate as is shown in figure \ref{fig:varSizeHist}. The distinction then comes from the distribution of the bulk
eigenvalues.

\begin{figure}[H]
  \centerline{
    \includegraphics[width=0.6\textwidth]{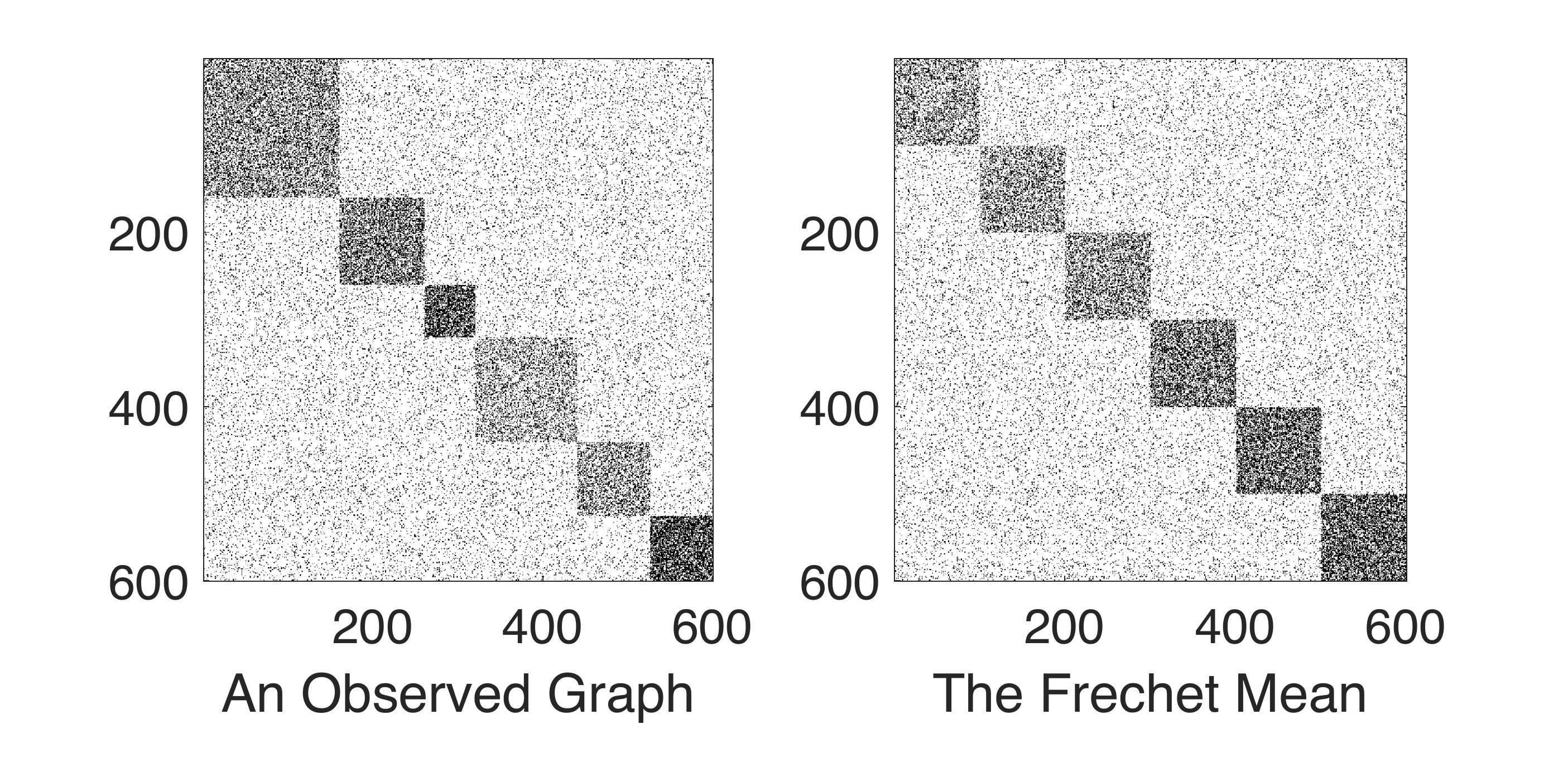}
    }
  \caption{Visual comparison between an observation from $\mu_3$ and the approximate Fr\'echet mean of $\mu_3$}
  \label{fig:varSizeFM}
\end{figure}

\begin{figure}[H]
  \centerline{
    \includegraphics[width=0.6\textwidth]{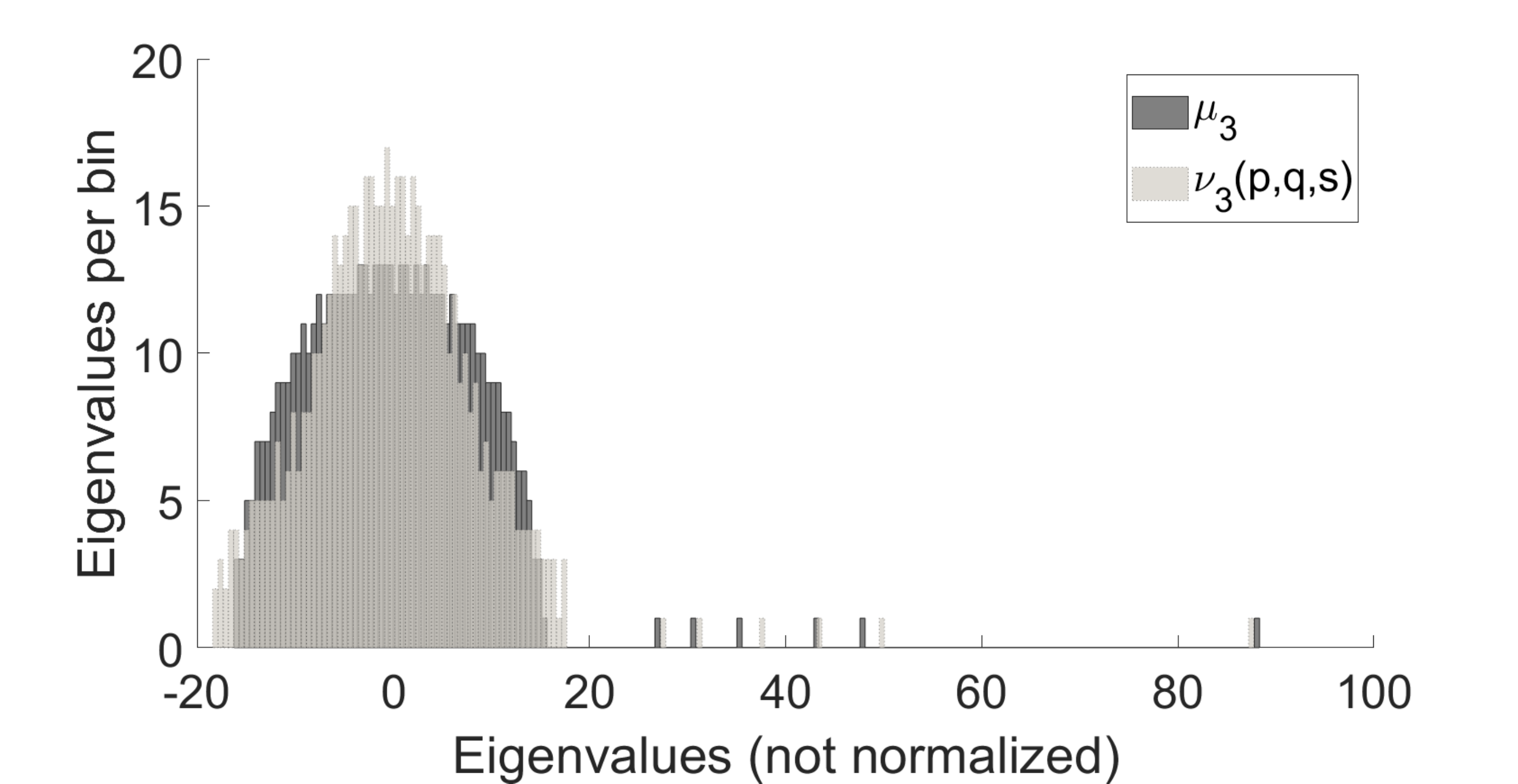}
    }
  \caption{Average histogram of spectra from $M_3$ overlayed by histogram of spectra from $\F(\nu_3)$}
  \label{fig:varSizeHist}
\end{figure}
Both of these figures suggest further research into the effect of community size and community strength on the extremal eigenvalues of a stochastic block model. It is worth noting that the extremal eigenvalues of stochastic block models are not solely dictated by the vector of community strengths $\bm p$ which suggests there exists sets  of parameters $(\bm p_1,q_1,\bs_1)$ and $(\bm p_2,q_2,\bs_2)$ such that
$(\bm p_1,q_1,\bs_1) \neq (\bm p_2,q_2,\bs_2)$ but $\sigma_c(\F(\mu_{\SBM}(\bm p_1,q_1,\bs_1))) = \sigma_c(\F(\mu_{\SBM}(\bm p_2,q_2,\bs_2))).$
This is noteworthy as it implies that within the class of stochastic block models with variable community sizes, the solution to (\ref{eqn:FM-D}) is not unique. If we instead had utilized a distance that measured the distribution of the bulk eigenvalues a unique solution would likely exist though this idea is not explored in this paper.

\subsection{Small World Empirical Fr\'echet Mean} 
The probability measure in this section is associated with the small world ensemble where the number of connected nearest
neighbors is $K = 22$ and the probability of rewiring is $\beta = 0.7$. Figure \ref{fig:swFM} gives a visualization of the
approximation performed by the stochastic block model ensemble.

One interpretation of figure \ref{fig:swFM} is to think of the stochastic block model kernel as a blockwise constant
approximation of a kernel for $\mu_4$. This interpretation is related to the work done in \cite{OW14} which motivated much of
the theory in this paper.

In figure \ref{fig:swHist} we again overlay the histograms of the graphs we consider. Notice that the histogram of eigenvalues
of the Fr\'echet mean of the stochastic block model approximation, $\nu_4$ seemingly has eigenvalues outside of the bulk on the
left as well as to the right. One potential cause of this phenomenon could be that the large number of communities leads to a
slower convergence to the generalized ``semi-circle" law, see discussion of this in \cite{ACK15}, and a larger graph is needed
to get a better estimation of the histogram over the bulk.
\begin{figure}[H]
  \centerline{
  \includegraphics[width=0.6\textwidth]{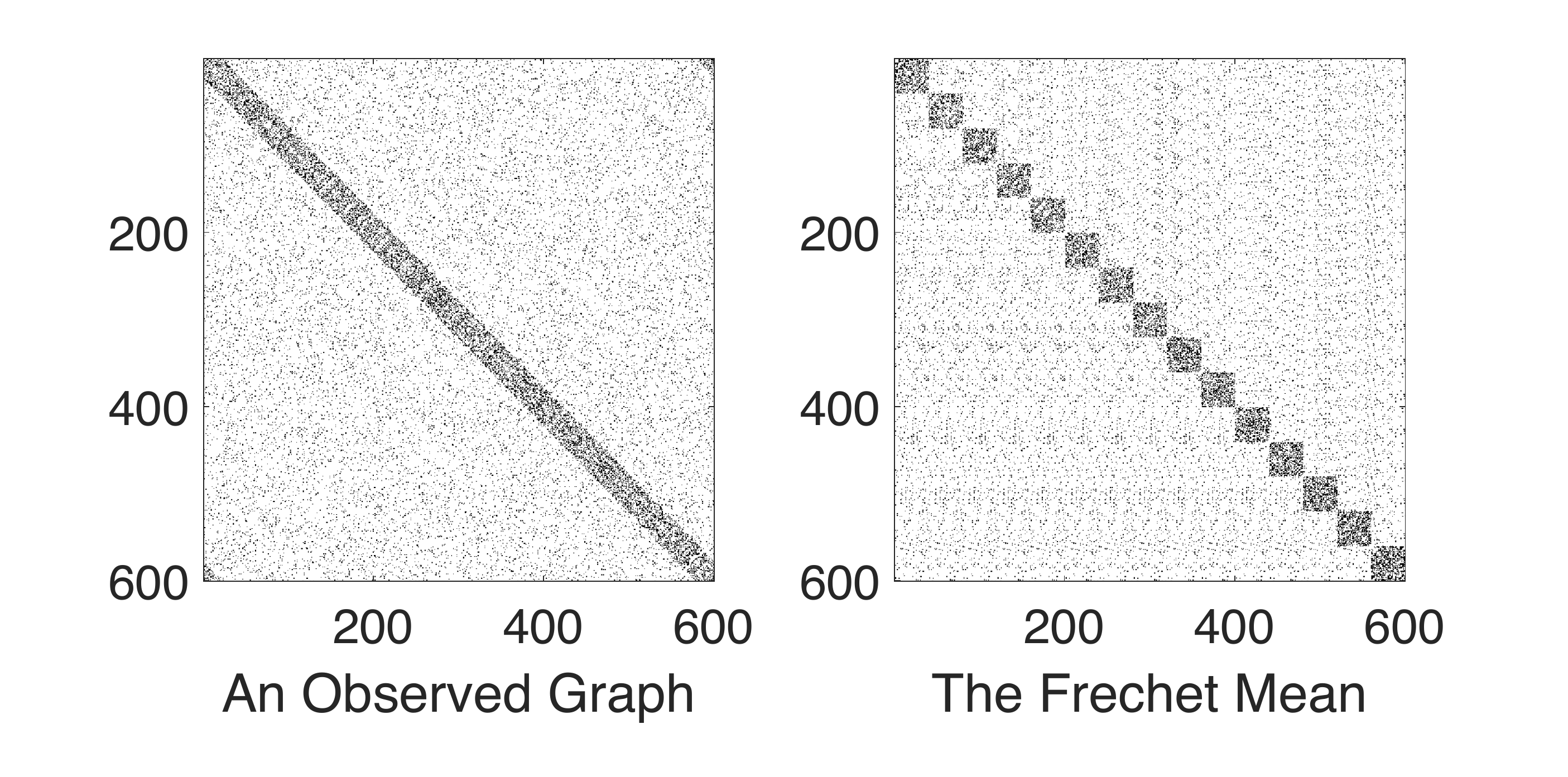}
  }
  \caption{Visual comparison between an observation from $\mu_4$ and the approximate Fr\'echet mean of $\mu_4$}
  \label{fig:swFM}
\end{figure}
\begin{figure}[H]
  \centerline{
    \includegraphics[width=0.6\textwidth]{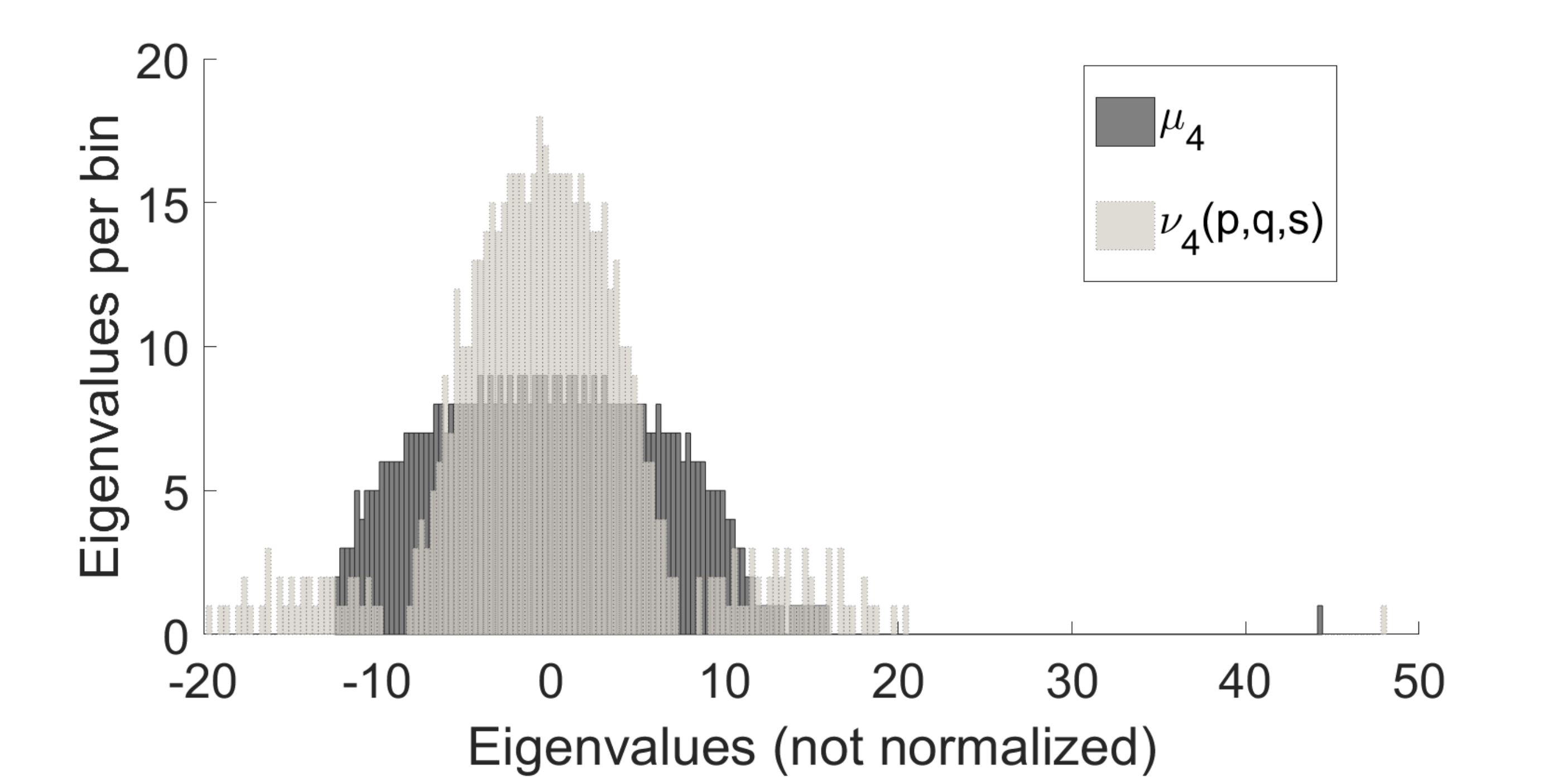}
    }
  \caption{Average histogram of spectra from $M_4$ overlayed by histogram of spectra from $h(\nu_4)$}
  \label{fig:swHist}
\end{figure}
\begin{figure}[H]
  \centerline{
    \includegraphics[width=0.6\textwidth]{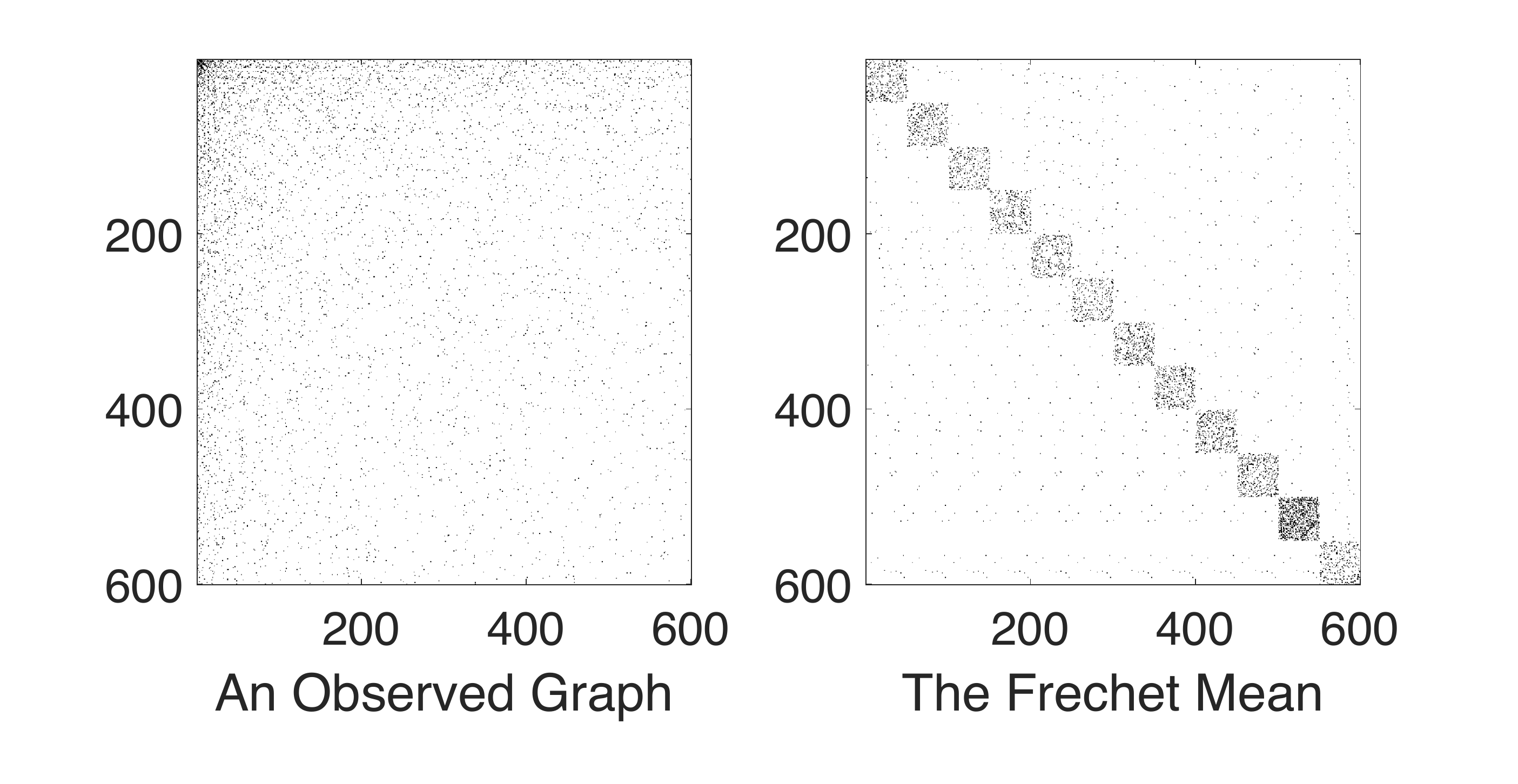}
    }
  \caption{Visual comparison between an observation from $\mu_5$ and the approximate Fr\'echet mean of $\mu_5$}
  \label{fig:baFM}
\end{figure}
Regardless of this we observe obvious visual similarities between the empirical Fr\'echet mean and an arbitrary graph in the
observed set $M_4$.
\begin{figure}[H]
  \centerline{
    \includegraphics[width=0.6\textwidth]{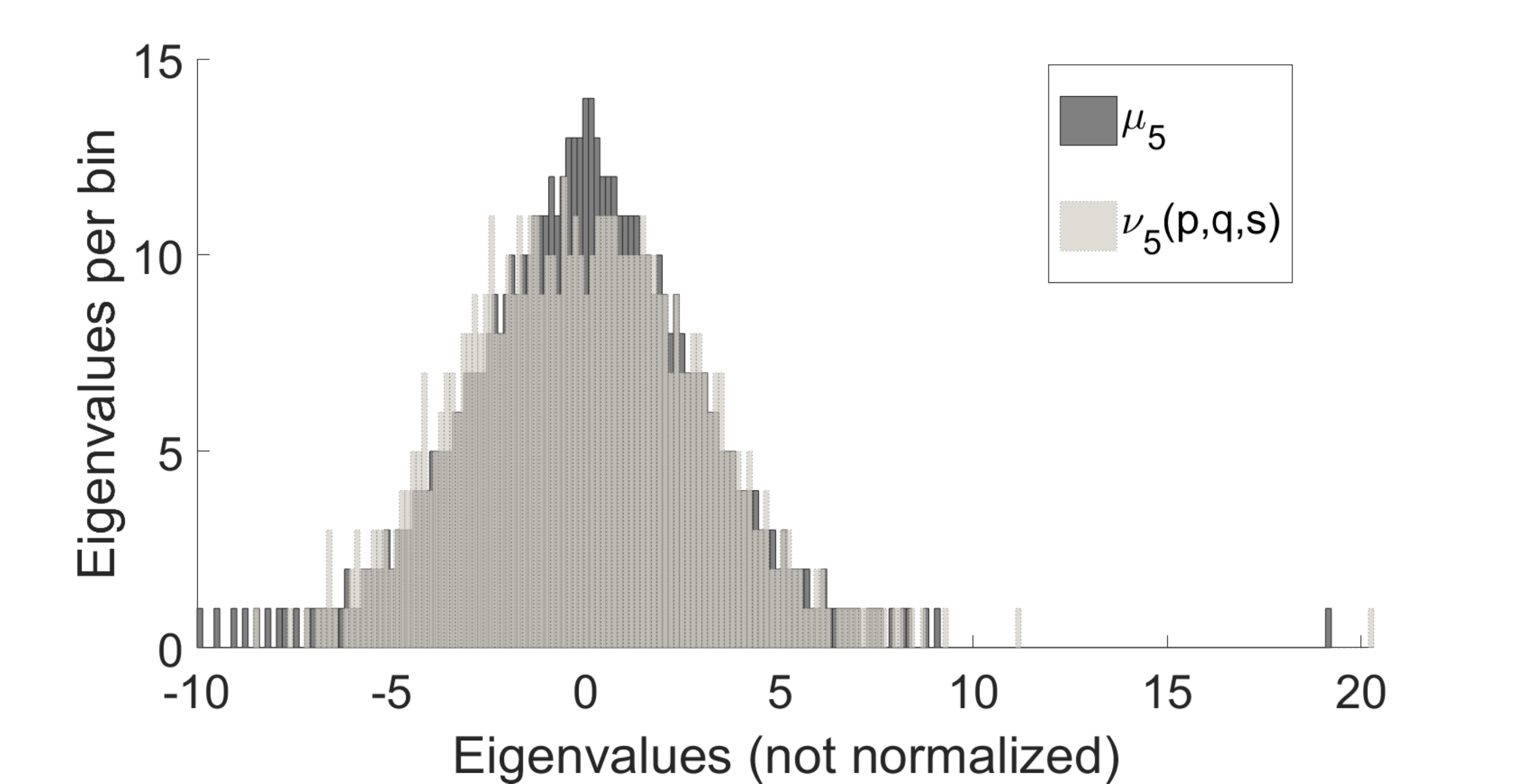}
    }
  \caption{Average histogram of spectra from $M_5$ overlayed by histogram of spectra from $\F(\nu_5)$}
  \label{fig:baHist}
\end{figure}
\subsection{Barabasi-Albert Empirical Fr\'echet Mean} 
The probability measure in this section is associated with a Barabasi-Albert ensemble. The initial graph for the ensemble is
fully connected on $m_0 = 5$ nodes and $m = 5$ edges were added at each step. In figure \ref{fig:baFM} we reorder the nodes
based on their degree for the Barabasi-Albert graph to get a better visual understanding of the similarities between an observed
graph and the Fr\'echet mean. The estimate for the number of communities from algorithm \ref{alg:DetC} is $c^* =12$ resulting in
$50$ vertices per community.

Figure \ref{fig:baHist} again depicts the alignment of the spectra from the approximate Fr\'echet mean with that of the average
spectra of the graphs from set $M_5$. Note the misalignment in the largest eigenvalues could be due to the finite graph
approximation. Recall all results hold in the limit of large graph size but throughout all of these experiments we are
approximating infinite graphs with finite graphs in addition to making the approximation by the stochastic block model ensemble.

\subsection{Application to regression
\label{subsec:ExpReg}}
This subsection is focused on performing a simplified experiment addressing the theory presented in section \ref{sec:Reg}. We first generate a synthetic data set of graphs by allowing the parameters of the stochastic block model to vary with time. For simplicity we hold $q$ and the non-zero entries of $\bs$ fixed as
\begin{equation}
\bs(t) = \begin{bmatrix}1/3 \\ 1/3 \\ 1/3 \\ 0 \\ \vdots \end{bmatrix}, q = 0.08.
\end{equation}
 
For $t \in [0,1]$ we let the non-zero entries of $\bm p$ vary linearly as
\begin{equation}
\bm p(t) = \begin{bmatrix} 0.1 + 0.1 t\\ 0.2 + 0.15t\\ 0.35 + 0.2t \\ 0 \\ \vdots \end{bmatrix}.
\end{equation}

For $T \sim unif(0,1)$, the distribution over $\cG$ is given as $\mu_{\SBM}(\bm p (T),q, \bs)$. For each sample from $unif(0,1)$ there is a corresponding sample from the stochastic block model. By construction we know the number of communities in the observed graphs will be constant at $c^* = 3$ dictating the number of non-zero entries of $\bm p$ we allow to vary when searching for a solution to equation (\ref{eqn:EFR}).\\
\begin{figure}[H]
  \centerline{
    \includegraphics[width=0.6\textwidth]{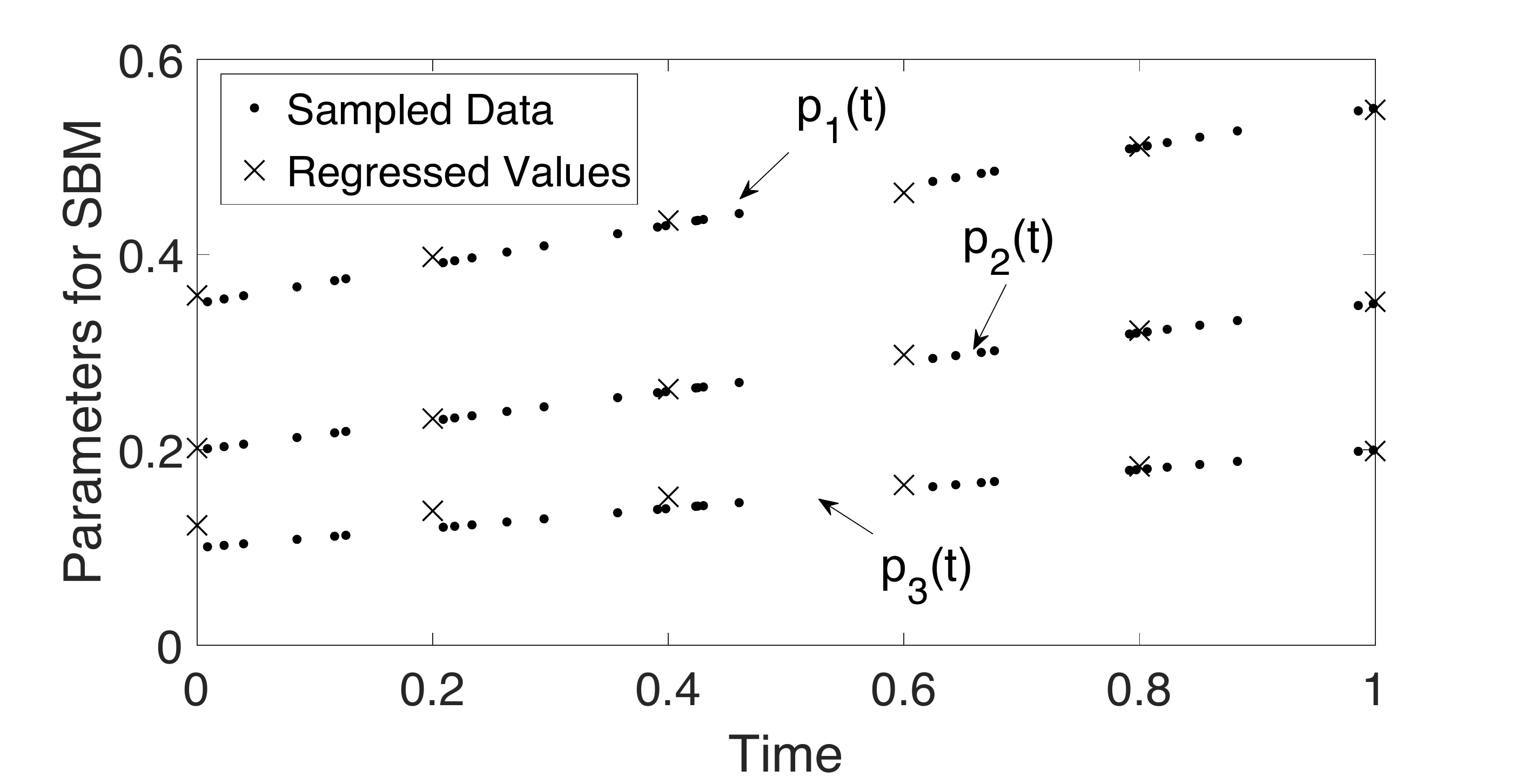}
    }
  \caption{Regression for Stochastic Block Models}
  \label{fig:FR}
\end{figure}
We take $N = 30$ samples for the sample set $M = \lbrace (t_i, G_i) \rbrace_{i=1}^{30}$ in the experiment. In an effort of visualization, since we are unable to plot a graph $G$ on the y-axis, we plot the estimated values $\bm p^*(t)$ at new time points. We expect to recover the lines that define $\bm p(t)$ for time values that were not sampled. Below we mark the estimated values of parameters for 6 graphs at the times $t \in \lbrace 0,0.2,0.4,0.6,0.8,1 \rbrace$ with a $\times$. A vertical line in figure \ref{fig:FR} indicates the three entries of $\bm p^*(t)$.

The true parameter values for the sampled graphs are displayed for comparison with the estimated values of the parameters at new
times. In practice, the true values are unknown to us without a significant amount of work and we would only be able to display
the recovered parameter values. A problem with this approach is that for each new time point, the corresponding Fr\'echet mean
must be found. This is incredibly costly since the evaluation of the objective involves solving a minimization problem at each
time point. We address this issue in forthcoming papers by performing regression on the recovered parameter values after
determining a few Fr\'echet means at select time points.

\section{Conclusion}
In the area of statistical analysis of for graph valued data, determining an average graph is a point of priority among
researchers. The standard practice in the field is to utilize the most central graph among the observed set of graphs as a
makeshift Fr\'echet mean. Throughout this paper, we have shown that when considering the metric $d_{A_c}$ it is possible to
determine an approximation to the empirical Fr\'echet mean given a dataset of sparse graphs.

How this approximate Fr\'echet mean is utilized is up to the discretion of the researcher however in section \ref{sec:Reg} we
explore one motivating idea that utilizes the Fr\'echet mean, termed Fr\'echet regression in the work in \cite{PM19}. This is
but one example of the utility of the Fr\'echet mean graph, another interesting application of this graph is to further push the
work in \cite{LOW20} which introduces a centered random graph model to capture the variance of a set of observations around a
mean graph.

Beyond the applicability of the Fr\'echet mean, theorem \ref{thm:Density} identifies a set of graphs that is dense, in the large
graph limit and with respect to $d_{A_c}$, in the set of sparse graphs. This result is useful in many respects as now we may
``project" (in some sense of the term) any large sparse graph onto the set of Fr\'echet means of stochastic block models and
begin to understand its structure as captured by the largest $c$ eigenvalues. This representation of a graph by the Fr\'echet
mean of a stochastic block model can be seen as a $2c+1$ dimensional approximately invertible embedding of a graph where the
embedding is the the $c$ non-zero entries of $\bm p$ and $\bs$ and the parameter $q$. This embedding furthermore allows for
natural analysis in the parameter space of the stochastic block model ensemble rather than analysis in $\cG$.
\section*{Acknowledgments}
F.G.M was supported by the National Natural Science Foundation (CCF/CIF 1815971).
\clearpage
\appendix
\label{sec:app}
\section{Classic Results}
\subsection{Universality of the Stochastic Block Model with respect to the 2-norm}
\begin{theorem} [\cite{OW14}]
  \label{thm:USBM2}
  Let $\mu_g$ be a probability measure with kernel function $g: [0,1]^2 \mapsto (0,1)$ such that $g$ is H$\ddot{o}$lder-continuous with expected density
  \begin{equation}
\E{\rho_n} = \Omega \left(\frac{\ln^3(n)}{n} \right).
\end{equation}

  Let $\mathcal{T}$ be the set of bijective functions on the interval $[0,1]$.
  For any $\varepsilon > 0$, there exists a kernel $f(x,y;\bm p, q, \bs)$ for a stochastic block model with the following properties:
  \begin{itemize}
  \item $\bs$ has $c$ non-zero entries
  \item $\bs(1) \geq \bs(i)$ for $i = 1,2,...,c$
  \item $\bs(i) = \bs(j)$ for $i,j = 2,...,c$
  \end{itemize}
  such that
  \begin{equation}
d_k(g(x,y),f(x,y;\bm p, q, \bs)) = \underset{\tau \in \mathcal{T}}{\text{inf }} \int \int_{(0,1)^2} \left| g(\tau(x),\tau(y)) - f(x,y;\bm p, q, \bs )\right|^2 dx dy \leq \varepsilon
\end{equation}

\end{theorem}
\begin{proof}
  The proof can be found in \cite{OW14}.
\end{proof}
The theorem  states that the class of kernels that generate sparse graphs in expectation can be approximated by
a kernel from the class of stochastic block models when we allow for the correct $\tau$ on the interval $[0,1]$. The choice of
$\tau$ is interpreted as a relabeling of the nodes.  
\subsection{Expected Value of the Largest Eigenvalues of a Stochastic Block Model}
\begin{theorem} [\cite{FFHL19}]
  \label{thm:ExEigs}
  Given kernel $f(x,y; \bm p, q, \bs) \in K_{\SBM}$ with probability measure $\mu_{\SBM}$ such that $\bs$ has $c$ non-zero
  entries. Let $G_{\mu_{\SBM}}$ be a random graph with adjacency matrix $\bA$. In the limit of large graph size, let
  $\blamb^{\bE} = \E{\sigma(\bA)}$. For $k = 1,...,c$, the $k^{th}$ largest eigenvalue, $\blamb^{\bE}(k) $, is the unique root
  of 
  \begin{equation}
f_k(z) = 1 + \blamb_{\bE}(k) \lbrace R(\bm v_k, \bm v_k,z) - R(\bm v_k, \bm V_{-k},z)[(D_{-k})^{-1}+ R(\bm V_{-k},\bm
V_{-k},z)]^{-1} \times R(\bm V_{-k},\bm v_k,z)\rbrace 
\end{equation}

  in the interval $[a_k,b_k]$ where
  \begin{equation}
a_k = \begin{cases} \frac{\blamb_{\bE}(k)}{1+c_0/2} \quad \blamb_{\bE}(k) > 0\\ \frac{1+c_0/2}{\blamb_{\bE}(k)} \quad
  \blamb_{\bE}(k) < 0\end{cases} \quad b_k = \begin{cases} \frac{1+c_0/2}{\blamb_{\bE}(k)} \quad \blamb_{\bE}(k) > 0\\
  \frac{\blamb_{\bE}(k)}{1+c_0/2}  \quad \blamb_{\bE}(k) < 0\end{cases}  
\end{equation}

  Asymptotically, \begin{equation}
f_k(z) = 1 + \blamb_{\bE}(k) R(\bm v_k, \bm v_k,z) + H.O.T.
\end{equation}

  Here $\blamb_{\bE} = \sigma(\E{\bA})$ and $\blamb_{\bE}(k)$ is the $k^{th}$ largest eigenvalue while $\bm v_k$ denotes the
  associated eigenvector. $\bm V$ is the $n \times c$ matrix of orthonormal eigenvectors of $\E{\bA}$ and $\bm V_{-k}$ is the
  submatrix of $\bm V$ with the $k^{th}$ column  and row removed. $\bm D$ is the diagonal matrix of eigenvalues of $\E{\bA}$
  organized in decreasing order. $\bm D_{-k}$ is the submatrix of $\bm D$ with the $k^{th}$ column and row removed.  
  \begin{equation}
R(M_1,M_2,z) = - \sum_{l = 0, l \neq 1}^L z^{-(l+1)} M_1^T \E{\bm W^l} M_2 
\end{equation}

  where $\bm W = \bA - \E{\bA}$ and $L = 4$ is typically sufficient and $c_0 \in (0,1)$.
\end{theorem}
\begin{proof}
  The proof can be found in \cite{FFHL19}.
\end{proof}
 This result gives an expression to determine the expected eigenvalues of graphs drawn from the stochastic block model
ensemble. As suggested in \cite{FFHL19}, the roots of $f_{k}(z)$ can be found using Newton's method.
\subsection{Convergence of Spectrum to Operator Spectrum}
\begin{theorem}
  \label{thm:SpecConv}
  Given a canonical kernel $f(x,y; \bm p, q, \bs) \in K_{\SBM}$ with probability measure $\mu_{\SBM}$ such that $\bs$ has $c$
  non-zero entries. Let $G_{\mu_{\SBM}}$ be a random graph with adjacency matrix $\bA$. For equispaced points $\lbrace \xi_i
  \rbrace_{i=1}^n$ in the interval $[0,1]$, define the expected adjacency matrix as 
  \begin{equation}
\E{\bA} = f(\xi_i,\xi_j; \bm p, q, \bs).
\end{equation}
  Denote the spectrum of the expected adjacency matrix by vector in $\R^n$ as
  \begin{equation}
\sigma(\E{\bA}) = \blamb_{\bE}
\end{equation}

  where $\blamb_{\bE}$ is sorted in descending order. Define the linear integral operator with kernel $f$, $L_f: L^2([0,1]) \mapsto L^2([0,1])$ as 
  \begin{equation}
L_f(g(x)) = \int_{0}^1 f(x,y; \bm p, s, \bs)g(y)dy.
\end{equation}

  Denote the spectrum of the linear integral operator as
  \begin{equation}
\sigma(L_f) = \blamb_{L_{f}}
\end{equation}

  where $\blamb_{L_{f}}$ is sorted in descending order and indexed by $i$.
  \begin{equation}
\lim_{n \to \infty} \frac{1}{\sqrt n} \blamb_{E}(i) = \blamb_{L_{f}}(i)
\end{equation}
\end{theorem}
\begin{proof}
This theorem is shown in each of \cite{S11, GC19, BCLSV12,J10} but has been adapted to the notations of this paper. The
interpretation of this result is that in the limit of large graph size, we may approximate the spectrum of the operator
associated with a stochastic block model by the eigenvalues of the discretized operator when appropriately normalized.  
\end{proof}
\subsection{Weyl-Lidskii}
\begin{theorem}
  \label{thm:SpecInc}
  Let $H$ be a self-adjoint linear operator on a Hilbert space $\cH$. Let $A$ be a bounded operator on $\cH$. Then
  \begin{equation}
\sigma(H + A) \subset \left \lbrace \lambda : dist(\lambda,\sigma(H)) \leq ||A|| \right \rbrace
\end{equation}
\end{theorem}
\begin{proof}
These are standard bounds that can be found in many good books on matrix perturbation theory (e.g., \cite{Stewart90}).
\end{proof}
\subsection{Extremal Eigenvalues of graphs from stochastic block models are normally distributed}
\begin{theorem}
  \label{thm:ExEigsNormDist}
  Given kernel $f(x,y; \bm p, q, \bs) \in K_{\SBM}$ with probability measure $\mu_{\SBM}$ such that $\bs$ has $c$ non-zero
  entries. Let $G_{\mu_{\SBM}}$ be a random graph with adjacency matrix $\bA$. In the limit of large graph size, the $c$ largest
  eigenvalues of $A$, denoted $\blamb(i)$, for $i = 1,...,c$ converges in distribution to a normal distribution
  \begin{equation}
\blamb(i) \sim N(m_i,v_i)
\end{equation}

  with mean $m_i$ and finite variance $v_i$. 
\end{theorem}
\begin{proof}
\end{proof}
We note that the mean $m_i$ is related to the eigenvalues of the expected adjacency matrix by way of theorem \ref{thm:ExEigs}.
\section{There exists a continuous map, $\tau_k$, from the extremal eigenvalues of $\sigma(\E{A})$ to the extremal
  eigenvalues of $\E{\sigma(A)}$} 
\begin{theorem}
  \label{thm:CtsMap}
  Given kernel $f(x,y; \bm p, q, \bs) \in K_{\SBM}$ with probability measure $\mu_{\SBM}$ such that $\bs$ has $c$ non-zero
  entries. Let $G_{\mu_{\SBM}}$ be a random graph with adjacency matrix $\bA$. For equispaced points $\lbrace \xi_i
  \rbrace_{i=1}^n$ in the interval $[0,1]$, define the expected adjacency matrix as 
  \begin{equation}
\E{\bA} = f(\xi_i,\xi_j; \bm p, q, \bs).
\end{equation}

  Denote the spectrum of the expected adjacency matrix by a vector in $\R^n$ as
  \begin{equation}
\sigma(\E{\bA}) = \blamb_{\bE}
\end{equation}

  where $\blamb_{\bE}$ is sorted in descending order.
  Denote the expected spectrum of the adjacency matrix by a vector in $\R^n$ as
  \begin{equation}
\E{\sigma(\bA)} = \blamb^{\bE}
\end{equation}

  where $\blamb^{\bE}$ is sorted in descending order. In the limit of large graph size, there exist $c$ continuous maps $\tau_k$ such that
  \begin{equation}
\tau_k(\blamb_{\bE}(k)) = \blamb^{\bE}(k) \quad k = 1,...,c.
\end{equation}

\end{theorem}
\begin{proof}
We need only prove for any general $k$. Let $f_k(z) = 1 + \blamb_{\bE}(k) R(\bm v_k, \bm v_k,z)$ from theorem
\ref{thm:ExEigs}. Note that $\blamb_{\bE}(k)$ depends continuously on the parameters of the kernel function, $\bm p, q, \bs$
since the spectrum is continuous with respect to the kernel function. Furthermore, $\bm v_k(i)$ depends continuously on the
parameters $\bm p, q, \bs$ for all $i = 1,...,c$. This implies that $f_k(z)$ is continuous in the interval $[a_k,b_k]$ where
$a_k,b_k$ are as defined in theorem $\ref{thm:ExEigs}$. Therefore $f_k(z)$ may be approximated by a polynomial $p_k(z)$ with
coefficients $\bm c$ such that $\bm c(i)$ depends continuously on the parameters $\bm p, q, \bs$. As a result, the roots of
$f_k(z)$ depend continuously on the parameters $\bm p, q, \bs$ since the roots of the polynomials $p_k(z)$ depend continuously
on the parameters.  
\end{proof}
\section{The spectra of linear integral operators are close if and only if the kernels are close}
\begin{theorem}
  \label{thm:EqNorms}
  Let $f$ be the kernel for probability measure $\mu$ and $f_{\SBM} \in K_{\SBM}$ be the kernel for probability measure
  $\mu_{\SBM}$. Let $L_f$ and $L_{f_{\SBM}}$ be the linear integral operators with kernels $f$ and $f_{\SBM}$ acting on
  $\mathcal{L}^2([0,1])$ defined below as 
  \begin{equation}
L_f(g(x)) = \int_0^1 f(x,y)g(y)dy
\end{equation}

  \begin{equation}
L_{f_{\SBM}}(g(x)) = \int_0^1 f_{\SBM}(x,y; \bm p, q, \bs)g(y)dy
\end{equation}

  Let $\varepsilon > 0$. Assume
  \begin{equation}
||f(x,y) - f_{\SBM}(x,y;\bm p, q, \bs)||_2 < \varepsilon \implies ||\sigma(L_f) - \sigma(L_{f_{\SBM}})||_2 < \frac{25}{4} \varepsilon
\end{equation}
\end{theorem}
\begin{proof}
The approach is to first show that the linear integral operator, $L_{k}$, with kernel $k(x,y) = f(x,y) - f_{\SBM}(x,y; \bm p, q,
\bs)$ has a norm controlled by $\varepsilon$. We then show that the eigenvalues of $L_{f}$ are close to the eigenvalues of
$L_{f_{\SBM}}$. We then bound the error by approximating the rate of decay in the tail of the spectrum and controlling the first
finite number of terms by $\varepsilon$. 

Let $\varepsilon > 0$ and assume that
\begin{equation}
||f(x,y) - f_{\SBM}(x,y;\bm p, q, \bs)||_2 < \varepsilon.
\end{equation}

Let $L_f$ and $L_{f_{\SBM}}$ be linear integral operators with kernels $f$ and $f_{\SBM}$ respectively acting on $\mathcal{L}^2([0,1])$ defined as

\begin{equation}
L_f(g) := \int_{[0,1]} f(x,y) g(y) dy
\end{equation}

\begin{equation}
L_{f_{\SBM}}(g) := \int_{[0,1]} f_{\SBM}(x,y; \bm p, q, \bs) g(y) dy
\end{equation}

Let $\blamb_{f} = \sigma(L_f)$ and $\blamb_{f_{\SBM}} = \sigma(L_{f_{\SBM}})$ be the spectra of the operators sorted in descending order of magnitude and indexed by $i \in I$, where $I$ is an arbitrary index set. Define $k(x,y) = f(x,y) - f_{\SBM}(x,y;\bm p, q, \bs)$ and the corresponding linear integral operator 
\begin{equation}
L_k(g) := \int_{[0,1]} k(x,y) g(y) dy.
\end{equation}

We first show that the norm of $L_k$ is small.
\begin{align}
  ||L_k||^2 &= \underset{||g|| = 1}{\text{sup }}||L_k(g)||_2^2\\
            &= \underset{||g|| = 1}{\text{sup }}||\int_0^1 k(x,y)g(y)dy||_2^2
\end{align}
For a fixed $x$ we have 
\begin{equation}
\int_{0}^1 k(x,y)g(y)dy \leq \sqrt{||k(x,y)||_2^2 ||g(y)||_2^2}
\end{equation}

by Cauchy-Schwarz so
\begin{align}
  ||L_k||^2 &\leq  \underset{||g|| = 1}{\text{sup }}||\sqrt{||k(x,y)||_2^2 ||g(y)||_2^2}||_2^2\\
            &= || \text{ }||k(x,y)||_2||_2^2\\
            &= \int_{0}^1 \left( \sqrt{\int_0^1 k(x,y)^2 dy}\right) dx.
\end{align}	
Let $h(x) = \int_0^1 k(x,y)^2 dy$ and define the set $X = \lbrace x \in [0,1]: h(x) < \varepsilon \rbrace$ . Then 
\begin{align}
  ||L_k||^2 &\leq \int_{0}^1 \left( \sqrt{h(x)}\right) dx\\
            &= \int_{X} \sqrt{h(x)} dx + \int_{[0,1]\backslash X} \sqrt{h(x)}dx\\
\end{align}
Since $f(x) = \sqrt{}$ is Lipschitz on $[\varepsilon,1]$ with Lipschitz constant $\frac{1}{2\sqrt{\varepsilon}}$ we have
\begin{align}
  ||L_k|| \leq \int_{X} \sqrt{h(x)} dx + \int_{[0,1]\backslash X} \frac{1}{2\sqrt \varepsilon }h(x)dx
\end{align} 
We then have
\begin{align}
  ||L_k|| &\leq \int_{X} \sqrt{h(x)} dx + \int_{[0,1]\backslash X} \frac{1}{2\sqrt \varepsilon }h(x)dx\\
          &= \int_{X} \sqrt{h(x)} dx + \int_{[0,1]\backslash X} \frac{1}{2\sqrt \varepsilon }h(x)dx\\
          &\leq \int_{X} \sqrt{h(x)} dx + \int_0^1 \frac{1}{2\sqrt \varepsilon }h(x)dx\\
          &<\sqrt \varepsilon + \frac{\sqrt \varepsilon}{2} = \frac 3 2 \sqrt{\varepsilon}
\end{align}
Thus $||L_k|| < \frac 3 2 \sqrt{\varepsilon}$ which shows the norm of $L_k$ is controlled by $\varepsilon$.

We want to show 
\begin{equation}
||\blamb_{f} - \blamb_{f_{\SBM}} || < \frac{25}{4} \sqrt{\varepsilon}.
\end{equation}

Note  \begin{equation}
L_{f} = L_{f_{\SBM}} + L_k.
\end{equation}

By the Weyl-Lidskii theorem,
\begin{equation}
|\blamb_{f}(i) - \blamb_{f_{\SBM}}(i) |< \frac 3 2 \sqrt{\varepsilon}
\end{equation}

since the norm of the operator $L_k$ is less than $\frac{3}{2} \sqrt{\varepsilon}$.
Furthermore, $||\blamb_f||_2 < \infty$ and $||\blamb_{f_{\SBM}}||_2 < \infty$ since the operators are compact. For $n^* > \frac{1}{\sqrt{\varepsilon}}$ we have that
\begin{align}
  ||\blamb_f - \blamb_{f_{\SBM}}||^2_2 &= \sum_{i=1}^\infty (\blamb_f (i) - \blamb_{f_{\SBM}}(i))^2\\
                                      &= \sum_{i=1}^{n^*} (\blamb_f (i) - \blamb_{f_{\SBM}}(i))^2 + \sum_{i = n^* + 1}^{\infty } (\blamb_f (i) - \blamb_{f_{\SBM}}(i))^2\\
                                      &\leq \sum_{i=1}^{n^*} (\frac 3 2 \sqrt{\varepsilon})^2 + \sum_{i = n^* + 1}^{\infty } (\blamb_f (i) - \blamb_{f_{\SBM}}(i))^2\\
                                      &= \frac 9 4 n^* \varepsilon + \sum_{i = n^* + 1}^{\infty } (\blamb_f (i) - \blamb_{f_{\SBM}}(i))^2\\
                                      &= \frac 9 4 n^* \varepsilon + \sum_{i = n^* + 1}^{\infty } \blamb_f^2 (i) - 2 \blamb_{f} \blamb_{f_{\SBM}}(i) + \blamb^2_{f_{\SBM}}(i)
\end{align}
Without loss of generality, we assume for $i > n^*$, both $\blamb_f(i) < \frac 1 i$ and $\blamb_{f_{\SBM}}(i) < \frac 1 i$. We therefore have a loose bound on the tails of the sequences of eigenvalues. Thus
\begin{align}
  ||\blamb_f - \blamb_{f_{\SBM}}||^2_2 &\leq \frac 9 4 n^* \varepsilon + \sum_{i = n^* + 1}^{\infty } \frac{1}{i^2} + 2 \frac{1}{i^2} + \frac{1}{i^2}\\
                                      &= \frac{9}{4} \varepsilon^{1/2} \varepsilon + 4 \varepsilon^{1/2}\\
                                      &= \frac{25}{4}\sqrt{\varepsilon}
\end{align}
This shows the forward direction. The backwards direction is a direct result from the theory of Hilbert-Schmidt operators. Let $f_1$ and $f_2$ be two kernels of linear integral operators defined as
\begin{equation}
L_{f_1} = \int_0^1 f_1(x,y) g(y) dy
\end{equation}

\begin{equation}
L_{f_2} = \int_0^1 f_2(x,y) g(y) dy
\end{equation}
 
where $f_1$ and $f_2$ are symmetric. Let $\sigma(L_{f_1}) = \blamb_1$ and $\sigma(L_{f_2}) = \blamb_2$ be the spectra of $L_{f_1}$ and $L_{f_2}$ respectively sorted in descending order. Assume that \begin{equation}
||\blamb_1 - \blamb_2||_2^2 < \varepsilon
\end{equation}

We want to show that 
\begin{equation}
||f_1 - f_2||_2^2 < \varepsilon
\end{equation}

Note $||L_{f_1}||_{HS} = \sum_{i=1}^\infty \blamb_1(i)^2$ and that $||L_{f_1}||_{HS} = ||f_1||_2^2$. As a consequence
\begin{equation}
||f_1||_2^2 =  \sum_{i=1}^\infty \blamb_1(i)^2.
\end{equation}

Similarly,
\begin{equation}
||f_2||_2^2 =  \sum_{i=1}^\infty \blamb_2(i)^2.
\end{equation}

Thus 
\begin{equation}
||f_1 - f_2||_2^2 = \sum_{i=1}^\infty (\blamb_1(i) - \blamb_2(i))^2 = ||\blamb_1 - \blamb_2||_2^2 < \varepsilon.
\end{equation}

\qed
\end{proof}
The work in \cite{OW14} states that we may approximate a certain class of kernel probability measures by stochastic block model
kernels, $K_{\SBM}$. We have shown that in doing so, we may also measure distances with respect to $d_A$ and maintain that the
spectra of the expected adjacency matrices remain close. We also will need to show that for a small change in the spectra of the
linear integral operators, the respective kernels remain close in the metric $d_k$. 
\section{The Fr\'echet mean of graphs from the stochastic block model is the expected spectrum}
\begin{theorem}
  \label{thm:FMES}
  Given  a canonical kernel $f(x,y; \bm p, q, \bs) \in K_{\SBM}$ for probability measure $\mu_{\SBM}$ such that $\bs$ has $c$
  non-zero entries. Let $G_{\mu_{\SBM}}$ be a random graph with adjacency matrix $\bA$. Let 
  \begin{equation}
\blamb^*_{\SBM} = \sigma(\underset{G \in M}{\text{argmin }}\E{d_{A_c}^2(G,G_{\mu_{\SBM}})})
\end{equation}

  \begin{equation}
\blamb^{\bE} = \E{\sigma(G_{\mu_{\SBM}})}
\end{equation}

  In the limit of large graph size, for $i = 1,...,c$
  \begin{equation}
\blamb^*_{\SBM}(i) = \blamb^{\bE}(i).
\end{equation}

\end{theorem}
\begin{proof}
  The extremal eigenvalues of the adjacency matrix follow a normal distribution by theorem \ref{thm:ExEigsNormDist}. The
  Fr\'echet mean of a normally distributed random variable is the same as its expected value. The conclusion follows. \qed
\end{proof}
\section{Density of Fr\'echet means of stochastic block models in the space of sparse graphs
\label{proof-density}}
\begin{theorem}
  \textbf{Density of Fr\'echet means of stochastic block models}\\
  \textit{(Theorem \ref{thm:Density} in the main paper)} \hfill \\
  Let $\cM_{\SBM}(\cG) \subset \cM$ denote the subset of distributions associated with stochastic block models. In the limit of
  large graph size, $\forall$ $G \in \mathcal{G}_s$, $\forall$ $\varepsilon > 0$, there exists $\mu_{\SBM} \in \cM_{\SBM}(\cG)$
  such that 
  \begin{equation}
d_{A_c}(G,\F(\mu_{\SBM})) < \varepsilon
\end{equation}

\end{theorem}
\begin{proof}
By theorem \ref{thm:USBM2}, for any kernel probability measure with the appropriate sparsity, we may find a kernel from
$K_{\SBM}$ that approximates it. By theorem \ref{thm:EqNorms}, if the kernels are close then the spectra of the induced linear
integral operators are close. By theorem \ref{thm:CtsMap} there is a continuous map from the spectra of the expected adjacency
matrix to the expected spectra of the adjacency matrix. By theorem \ref{thm:FMES} the Fr\'echet mean of the stochastic block
model is the graph that achieves $\E{\sigma(G_{\mu_{\SBM}})}$. Any sparse graph $G$ may be written as $\F(\mu)$ for a certain
probability measure $\mu$. Since $\mu$ may be approximated by $\mu_{\SBM}$ we may estimate $\F(\mu)$ as $\F(\mu_{\SBM})$.
\qed 
\end{proof}
\section{A sample statistic to approximate the Fr\'echet mean of a stochastic block model with high probability}
\begin{theorem}
  \label{thm:FiniteSampFM}
  Let $\mu_{\SBM} \in \cM_{\SBM}(\cG)$. Let $M = \lbrace G_i \rbrace_{i=1}^N$ be an iid sample distributed according to $\mu_{\SBM}$. Define
  \begin{equation}
S_N = \underset{G \in M}{\mathrm{argmin} \text{ }}\sum_{i=1}^N d_{A_c}^2(G,G_i)
\end{equation}

  In the limit of large system size, for every $\varepsilon$ and $\delta$, there exists an $N$ such that
  \begin{equation}
P(d_{A_c}(S_N - \F(\mu(\bm p,q;c))<\delta) > 1-\varepsilon
\end{equation}

\end{theorem}
\begin{proof}
By theorem \ref{thm:ExEigsNormDist} the $c$ largest eigenvalues of the stochastic block models are normally
distributed with finite variance. Therefore we only need to show this result for normal random variables in $\R^c$. Let $M =
\lbrace \bm x_i \rbrace_{i=1}^N$ be an iid sample from a normal distribution with mean $\bm \mu$ and covariance $\sigma$. Let
$S_N = \underset{\bm x \in M}{\text{argmin }}\sum_{i=1}^N ||\bm x_i - \bm x||.$  Note that the Fr\'echet mean of a normal
distribution is its expectation. Condiser,
\begin{equation}
P(\vert \vert S_N - \bm \mu \vert \vert < \delta) = 1 - P(\vert \vert S_N - \bm \mu \vert \vert > \delta).
\end{equation}

$P(\vert \vert S_N - \bm \mu \vert \vert > \delta)$ states that every element in $M$ is at least $\delta$ away from the mean. Let $p = P(||\bm x_1 - \bm \mu|| > \delta)$. Then 
\begin{equation}
 1 - P(\vert \vert S_N - \bm \mu \vert \vert > \delta) = 1-P(||\bm x_1 - \bm \mu|| > \delta)^N = 1-p^N.
\end{equation}

Since $0 \leq p < 1$ then for any $\varepsilon$ we need only to pick $N$ such that 
\begin{equation}
1-p^N > 1- \varepsilon.
\end{equation}
Taking $N > \frac{\ln(\varepsilon)}{\ln p}$ guarantees the result. Note that the dependence on $\delta$ is suppressed in the
probability $p$ which implicitly depends on $\delta$. We have shown the conclusion for normally distributed random variables and
since the extremal eigenvalues of stochastic block models are normally distributed, by theorem \ref{thm:ExEigsNormDist}, the
conclusion follows.
\qed
\end{proof}

\bibliographystyle{acm}

\end{document}